\documentclass[aps,pra,twocolumn,letterpaper,superscriptaddress,floatfix,nofootinbib,longbibliography]{revtex4-1}

\usepackage{times}
\usepackage{algorithm}
\usepackage{algorithmic}

\usepackage[braket,qm]{qcircuit}
\usepackage{amsmath}
\usepackage{amssymb}
\usepackage{bm}
\usepackage{epsfig}
\usepackage{epstopdf}
\usepackage{subfigure}
\usepackage{array}
\usepackage{multirow}
\usepackage{eucal}
\usepackage{makecell}
\usepackage{fnbreak}

\usepackage{amsthm}

\newtheorem{theorem}{Theorem}

\newtheorem{proposition}{Proposition}
\newtheorem{lemma}{Lemma}

\usepackage{relsize}
\usepackage{graphicx}
\usepackage{caption}
\usepackage{epsfig}
\usepackage{epstopdf}
\usepackage{subfigure}
 
\usepackage{array}

\DeclareMathOperator{\tr}{tr}

\usepackage{color}
\usepackage{url}
\usepackage{textcomp} 
\usepackage{enumitem}
\setlist{noitemsep,leftmargin=*}
\usepackage{hyperref}
\hypersetup{
colorlinks=true,
urlcolor= blue,
citecolor=blue,
linkcolor= blue,
}

\begin{document}
\title{Variational quantum Gibbs state preparation with a truncated Taylor series}
\author{Youle Wang}
\affiliation{Institute for Quantum Computing, Baidu Research, Beijing 100193, China}
\affiliation{Centre for Quantum Software and Information, University of Technology Sydney, NSW 2007, Australia}
\author{Guangxi Li}
\affiliation{Institute for Quantum Computing, Baidu Research, Beijing 100193, China}
\affiliation{Centre for Quantum Software and Information, University of Technology Sydney, NSW 2007, Australia}
\author{Xin Wang}
\thanks{wangxin73@baidu.com} 
\affiliation{Institute for Quantum Computing, Baidu Research, Beijing 100193, China}

\begin{abstract}
The preparation of quantum Gibbs state is an essential part of quantum computation and has wide-ranging applications in various areas, including quantum simulation, quantum optimization, and quantum machine learning. In this paper, we propose variational hybrid quantum-classical algorithms for quantum Gibbs state preparation. We first utilize a truncated Taylor series to evaluate the free energy and choose the truncated free energy as the loss function. Our protocol then trains the parameterized quantum circuits to learn the desired quantum Gibbs state. Notably, this algorithm can be implemented on near-term quantum computers equipped with parameterized quantum circuits. By performing numerical experiments, we show that shallow parameterized circuits with only one additional qubit can be trained to prepare the Ising chain and spin chain Gibbs states with a fidelity higher than 95\%. In particular, for the Ising chain model, we find that a simplified circuit ansatz with only one parameter and one additional qubit can be trained to realize a 99\% fidelity in Gibbs state preparation at inverse temperatures larger than $2$.
\end{abstract}

\maketitle

 \section{Introduction}
\label{SEC:INTRODUCTION}
Quantum state preparation is an integral part of quantum computation. In the near future, quantum computers will become quantum state preparation factories for various tasks. One central task in quantum state preparation is to prepare the quantum Gibbs or thermal state of a given Hamiltonian. The reason is that quantum Gibbs states not only can be used to study many-body physics but also can be applied to quantum simulation~\cite{childs2018toward}, quantum machine learning~\cite{kieferova2017tomography,biamonte2017quantum}, and quantum optimization~\cite{somma2008quantum}. In particular, sampling from well-prepared Gibbs states of Hamiltonians can be applied in solving combinatorial optimization problems~\cite{somma2008quantum}, solving semi-definite programs~\cite{Brandao2017}, and training quantum Boltzmann machines~\cite{Kieferova2016}. 

The preparation of the desired initial state is a difficult problem in general. It is well-known that finding ground states of Hamiltonians is QMA-hard~\cite{Watrous2012a}. Gibbs state preparation at arbitrary low temperature could be as hard as finding the ground state~\cite{aharonov2013quantum}. There are various methods to do this using classical~\cite{terhal2000problem,poulin2009sampling,temme2011quantum,kastoryano2016quantum,brandao2019finite} or quantum computers.
Especially, these methods use quantum computing techniques, including quantum rejection sampling~\cite{Wiebe2016}, quantum walk~\cite{Yung2012}, dynamics simulation~\cite{Kaplan2017,Riera2012,Motta2019}, dimension reduction~\cite{Bilgin2010a}. Although in the worst case, the costs of these methods require could be exponential in expectation, they can be efficient when some conditions are satisfied. Such as the ratios between the partition functions of the infinite temperature states and the Gibbs state is at most polynomially large~\cite{poulin2009sampling}, and the gap of the Markov chain that describes the quantum walk is polynomially small~\cite{Yung2012,van2020quantum}.
However, these methods require complex quantum subroutines such as quantum phase estimation, which are costly and hard to implement on near term quantum computers.

The main goal of our work is to prepare quantum Gibbs states using near term quantum devices. In contrast to the previous works, we aim to reduce the task's quantum resources, such as qubit and gate counts, circuit size. 
To this end, one feasible scheme is to take advantage of variational quantum algorithms (VQAs)~\cite{mcclean2016theory}, as VQAs have recently been gaining popularity with applications in many areas \cite{xu2021variational,bravo2020variational,huang2019near,larose2019variational,Cerezoa,Wang2021,peruzzo2014variational,nakanishi2019subspace,Wang2020a,Sharma2020,Li2021,Chen2020a}. VQAs are a class of hybrid quantum-classical algorithms that involve optimizing a loss function depending on the circuit's parameters, where the loss function is evaluated on quantum devices, and parameters adjustments are outsourced to the classical devices. This strategy succeeds in reducing the resources by using shallow quantum circuits. Several methods based on VQAs have already been proposed~\cite{Islam2015,verdon2019quantum,Wu2019b,martyn2019product,mcardle2019variational,Yuan2019,Motta2019,Chowdhury2020} to prepare Gibbs states. For instance, Wu and Hsieh~\cite{Wu2019b} proposed a variational approach by using R{\' e}nyi entropy estimation~\cite{Islam2015} and thermofield double states, Yuan et al.~\cite{Yuan2019} discussed the application of imaginary time evolution to Gibbs state using parameterized circuits, and Chowdhury et al.~\cite{Chowdhury2020} proposed entropy estimation method using tools such as quantum amplitude estimation and linear combination of unitaries.

{In this paper, we propose variational quantum algorithms to tackle the problem of Gibbs state preparation on near-term quantum hardware. In general, a Gibbs state of the system's Hamiltonian $H$ at the inverse temperature $\beta=1/k_BT$ is defined as $\rho_{G}=e^{-\beta H}/\tr(e^{-\beta H})$, where $k_B$ is the Boltzmann constant, and $T$ is the system's temperature. To prepare $\rho_G$, we utilize the variational principle of the system's free energy~\cite{reif2009fundamentals}, which states that the Gibbs state minimizes the free energy. Let the system's state be $\rho$, then the free energy is given by $\mathcal{F}(\rho)=\tr(H\rho)-\beta^{-1}S(\rho)$, where $S(\rho)$ denotes the von Neumann entropy. Then $\rho_G$ is the global minimum of the functional $\mathcal{F}(\rho)$. In our approach, we use parameterized quantum circuits (PQCs) to prepare quantum states. Particularly, the PQCs are assumed to admit enough expressiveness to prepare the desired Gibbs states or a state very close to it. Throughout the paper, the state prepared via PQC is denoted by $\rho(\bm\theta)$. With these notations, the variational principle could be formulated as 
\begin{align}
\rho_{G}\approx\text{argmin}_{\bm\theta} \mathcal{F}(\rho(\bm\theta)).
\end{align}
Our goal is to find the optimal parameters that minimize the free energy, given a suitable PQC. Viewed from this point, we could prepare a state very close to the desired state $\rho_G$.
}

{The main challenge of minimizing the free energy comes from the entropy estimation, which is well-known to be difficult \cite{Gheorghiu}. To overcome this challenge, we truncate the Taylor series of the entropy at order $K$ and set the truncated free energy as the loss function in our variational quantum algorithms. Explicitly, the loss function is represented as a linear combination of system's energy and higher-order state overlaps. The energy and state overlaps are efficiently estimated via certain quantum gadgets, which could be performed on near-term quantum hardware. Hence, our algorithms for Gibbs state preparation are expected to be practical.
}

{We also study the effectiveness of our variational quantum algorithms in Gibbs state preparation. Specifically speaking, we theoretically show that fidelity between the output state of our approach and the target Gibbs state increases as truncation order $K$ increases. To demonstrate the efficacy of our approach, we focus on the loss function with truncation order $2$. We numerically find out that using loss function with order two suffices to prepare high-fidelity Gibbs states of several many-body Hamiltonians at selected low temperatures. 
As shown in the numerical results, the final fidelity could reach at least 95\% for the Ising chain and $XY$ spin-$\frac{1}{2}$ chain models, respectively. Particularly, the preparation of Ising chain Gibbs states achieves a fidelity above 99\%. Besides, we show that our approach applies to prepare Gibbs states at high temperatures. Overall, the numerical results imply that we could prepare high-fidelity Gibbs states of certain many-body Hamiltonians using low-order loss functions, strengthening the feasibility of our approach on near-term quantum hardware.}

\textbf{Organization. }
Our paper proceeds as follows. In Sec.~\ref{SEC:RESULT}, we describe our main results of this paper, including the variational algorithm for Gibbs state preparation, error analysis, and the analytical gradient of the loss function for optimization. In Sec.~\ref{SEC:EXPERIMENT}, we numerically demonstrate the effectiveness of our algorithm, especially at lower and higher temperatures, via targeting the Gibbs state preparation of Ising model and XY spin-1/2 model. Then we conclude in Sec.~\ref{SEC:DISCUSSION}. Finally, detailed proofs are supplemented in the Appendix.

\section{Main results}
\label{SEC:RESULT}
\subsection{Hybrid quantum-classical algorithm for Gibbs state preparation}

The Gibbs state for a quantum Hamiltonian $H$ is defined as the density operator
\begin{align}
\rho = \frac{\exp(-\beta H)}{\tr(\exp(-\beta H))}.
\end{align}
We recall that the free energy of a system described by a density operator $\omega$ is given by 
\begin{align}\label{eq:free_energy}
\mathcal{F}(\omega) = \tr(\omega H) - \beta^{-1}S(\omega),
\end{align}
where $\beta=(k_BT)^{-1}$ is the inverse temperature of the system, $k_B$ is the Boltzmann’s constant, and $S(\rho):=-\tr\rho\ln\rho$ is the von Neumann entropy of $\rho$. As the Gibbs state minimizes the free energy of the Hamiltonian $H$, it holds that
\begin{align}
\rho = \text{argmin}_{\omega} \mathcal F (\omega).
\end{align} 
Therefore, if we could generate parameterized quantum states $\omega(\bm\theta)$ and find a way to measure or estimate the loss function $\tr(\omega H) - k_BT \cdot S(\omega)$, then one could design variational algorithms via the optimization over $\bm\theta$~\cite{Wu2019b,Chowdhury2020}.

However, determining the von Neumann entropy of a quantum state on near-term quantum devices is quite challenging. Existing quantum algorithms for estimating certain entropic quantities~\cite{Subramanian2019a,Li2019b,Chowdhury2020} are not suitable for our purpose of using near-term quantum devices since they either have an explicit polynomial dependence on the Hilbert-space dimension of quantum system or require certain oracle assumption.  Recently, the authors of \cite{Chowdhury2020} propose a procedure for estimating the von Neumann entropy and free energy that uses tools such as quantum amplitude estimation~\cite{Brassard2002}, density matrix exponentiation~\cite{Low2016}, and techniques for approximating operators by a linear combination of unitaries. 

To design a suitable and efficient variational quantum algorithm for near-term quantum devices, we design a loss function in the similar spirit of free energy. This loss function utilizes the truncated Taylor series and, in particular, can be efficiently estimated on near-term quantum devices. {Specifically, the core of our idea is to truncate the von Neumann entropy as a linear combination of higher-order state overlaps, i.e., $\tr(\rho^{k})$, and estimate each $\tr(\rho^{k})$ via quantum gadgets, e.g., Swap test, respectively. Next, we employ a classical optimizer to minimize the loss function via tuning the parameters and then use the optimized circuit to prepare the target Gibbs state. Note that the subroutine of loss evaluation occurs on the quantum devices, and the procedure of optimization is entirely classical. Then, classical optimization tools such as gradient-based or gradient-free methods can be employed in the optimization loop.}

In general, the variational quantum circuit contains a series of parameterized single-qubit Pauli rotation operators and CNOT/CZ gates alternately \cite{Benedetti2019a}. Here we follow this circuit pattern and mainly use Pauli-Y rotation operators and CNOT gates. For the optimization part, a variety of approaches have been proposed to optimize such variational quantum circuits, including Nelder-Mead \cite{Guerreschi2017,Verdon2017}, Monte-Carlo~\cite{Wecker2016}, quasi-Newton~\cite{Guerreschi2017}, gradient descent \cite{Wang2018f}, and Bayesian methods. Here we choose ADAM  \cite{kingma2014adam} as our gradient-based optimizer in the numerical experiments.

{For simplicity, the overall hybrid algorithm with the two-order loss function is given in Algorithm~\ref{alg:rftl}, and a picture for illustration is depicted in Fig.~\ref{fig:Gibbs prepare}. Besides, the variational quantum algorithm for general truncation order $K$ is deferred to the Appendix~\ref{appendix:detals}.}

\begin{figure}[t]
\begin{algorithm}[H]
\caption{Variational quantum Gibbs state preparation with truncation order $2$}
\label{alg:rftl}
\begin{algorithmic}[1]
\STATE Choose the ansatz of unitary $U(\bm\theta)$, tolerance $\varepsilon$, truncation order $2$, and initial parameters of $\bm\theta$;
\STATE Compute coefficients $C_{0}$, $C_{1}$, $C_{2}$ according to Eq.~\eqref{EQ:13}. 
\STATE Prepare the initial states $\ket{00}$ in registers $AB$ and apply  $U(\bm\theta)$.
\STATE Measure and compute  $\tr (H\rho_{B_1})$ and compute the loss function $L_1 =  \tr (H\rho_{B_1} )$;
\STATE Measure and compute $\tr (\rho_{B_2}\rho_{B_3})$ via Destructive Swap Test and compute the loss function $L_2 =-\beta^{-1}C_{1}\tr(\rho_{B_{2}}\rho_{B_{3}})$; 
\STATE Measure and compute $\tr(\rho_{B_{4}}...\rho_{B_{6}})$ via higher-order state overlap estimation and compute the loss function $L_{3}=-\beta^{-1}C_{2}\tr(\rho_{B_{4}}...\rho_{B_{6}})$.
\STATE Perform optimization of $\mathcal{F}_2(\bm\theta) = \sum_{k=1}^{3}L_k-\beta^{-1}C_{0}$ and update parameters of $\bm\theta$;
\STATE Repeat 3-7 until the loss function $\mathcal F_2(\bm\theta)$ converges with tolerance $\varepsilon$;
\STATE Output the state $\rho^{out} = \tr_A U(\bm\theta)\op{00}{00}_{AB}U(\bm\theta)^\dagger$.
\end{algorithmic}
\end{algorithm}
\end{figure}

{Clearly, Algorithm 1 can be efficiently implemented on near-term quantum devices since the estimation of loss function $\mathcal F_2$ only requires measuring the expected value ${\left\langle H \right\rangle _\rho }$, the purity or the state overlap $\tr(\rho^2)$, and the higher-order state overlap $\tr(\rho^3)$}.
To compute the state overlap, one approach is to utilize the well-known Swap test~\cite{Buhrman2001,Gottesman2001a}, which has a simple physical implementation in quantum optics~\cite{Ekert2002,Garcia-Escartin2013} and can be experimentally implemented on near-term quantum hardware~\cite{Islam2015,Patel2016,Linke2018}. For instance, we could use a variant version of the Swap test (see Fig.~\ref{figure:destructive swap}), named destructive Swap
test~\cite{Garcia-Escartin2013,cincio2018learning}. Compared to the general Swap test, destructive Swap test is more practical on near term devices, since it is ancilla-free and costs less circuit depth and the number of the gates. {Using the circuit in Fig.~\ref{figure:destructive swap}, the quantity $\tr(\rho^2)$ is expected to be estimated on near-term quantum hardware.}

\begin{figure}[htb]
    \centering
    \includegraphics[width=0.25\textwidth]{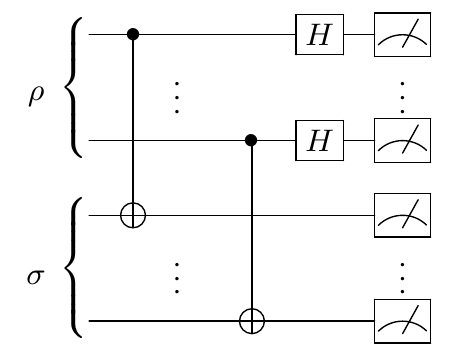}
    \caption{Quantum circuit for implementing Destructive Swap Test. In the circuit, two states $\rho$ and $\sigma$ are prepared at different registers. Then CNOT and Hadamard gates are performed as shown. The state overlap can be estimated via post-processing.}
    \label{figure:destructive swap}
\end{figure}

{Regarding higher order state overlaps, e.g., $\tr(\rho^{3})$, there are methods using the similar circuit to that of the destructive Swap test~\cite{subacsi2019entanglement}, whose depth is only 2. For more information, please refer to~\cite{subacsi2019entanglement}.  We can also use the qubit-efficient circuit proposed by Yirka et al. \cite{Yirka2020} to compute $\tr(\rho^k)$, for larger $k$. Particularly, the circuit only uses a constant number of qubits, where the key is that some subset of qubits can be reset in the course of quantum computation. For convenience, we call this method the higher-order state overlap estimation and provide an example for computing $\tr(\rho^3)$ in Fig.~\ref{figure:qubit-efficient swap}. We also refer interested readers to \cite{Yirka2020} for more details on qubit-efficient algorithms for computing $\tr(\rho^k)$. Hence, using these qubit-efficient quantum circuits will significantly circumvent our approach's resource requirements for computing $\tr(\rho^k)$ for $k\geq 3$ and enable our approach to be implementable on NISQ computers. }

\begin{figure}[htb]
    \centering
    \includegraphics[width=0.45\textwidth]{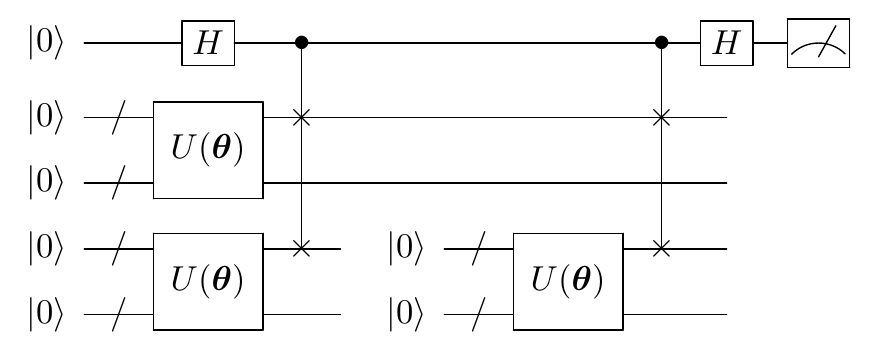}
    \caption{Quantum circuit for computing $\tr(\rho^3)$. In the circuit, the $U(\bm\theta)$ denotes the state preparation circuit, and $H$ denotes the Hadamard gate. Four registers are used to prepare states by $U(\bm\theta)$, and one ancillary qubit is used to perform the controlled swap operator. The qubit reset occurs on the bottom two registers, where the break in the wire means the reset operation. Notably, the state on the bottom two registers are first implemented with a circuit $U(\bm\theta)$ and controlled swap operator and then reset to state $\ket{0}$. Again, $U(\bm\theta)$ and controlled swap operator are performed on the bottom registers. Finally, $\tr(\rho^3)$ can be obtained via post-processing the measurement results.}
    \label{figure:qubit-efficient swap}
\end{figure}

 \begin{figure*}[htb]
\centering
\includegraphics[width=15cm]{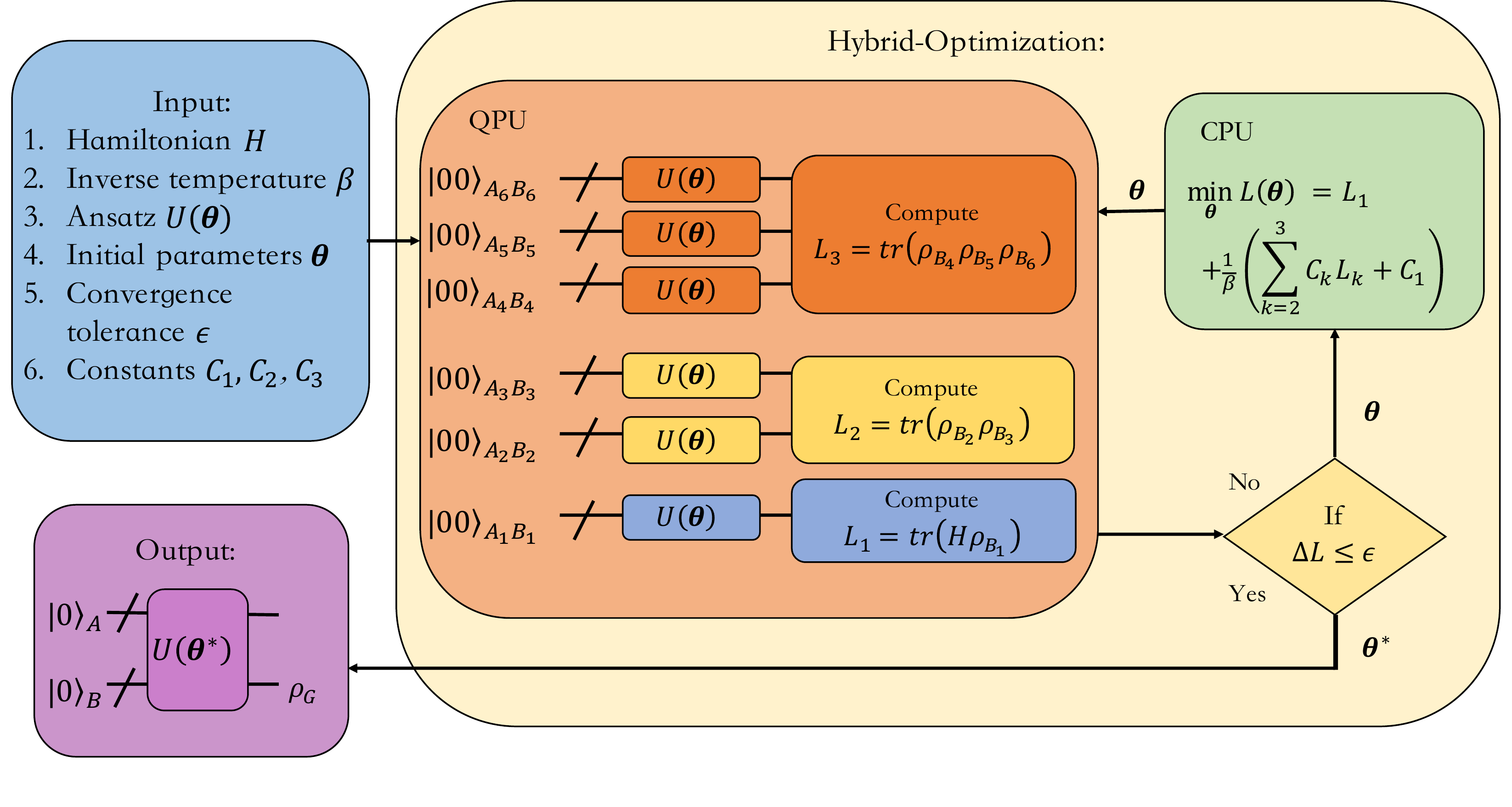}
\caption{Schematic representation of the variational quantum Gibbs state preparation with truncation order $2$. First, we prepare the Hamiltonian $H$ and inverse temperature $\beta$ and then send them into the Hybrid Optimization. Second, we choose an ansatz and employ it to evaluate the loss function $L_{1},L_{2},L_{3}$ on quantum devices. Then we calculate the difference $\Delta\mathcal{F}_{2}(\bm\theta)$ by using $L_{1},L_{2},L_{3}$. Next, if the condition $\Delta \mathcal{F}_{2}\leq\epsilon$ is not satisfied, then we perform classical optimization to update parameters $\bm\theta$ of the ansatz and return to the loss evaluation. Otherwise, we output the current parameters $\bm\theta^*$, which could be used to prepare Gibbs state $\rho_{G}$ via $U(\bm\theta)$. Here in the quantum device, registers $A_2,B_2, A_3, B_3$ are used to evaluate $\tr{(\rho_{B_2}\rho_{B_3})}$ and registers $A_4, B_4,\ldots, A_6,B_6$ are used to evaluate $\tr{(\rho_{B_4}\rho_{B_5}\rho_{B_6})}$. }
\label{fig:Gibbs prepare} 
\end{figure*}

\subsection{Performance analysis of our method}
In this section, we analyze the performance of our variational algorithm. {Specifically, we first define a formal optimization problem that aims to find the global minimum of the truncated free energy. Second, we show that the prepared state has a higher overlap with the desired Gibbs state, using higher-order loss functions in our approach.}

\paragraph{Loss function}
In our algorithm (Algorithm~\ref{alg:rftl:K}), the $K$-truncated-free energy $\mathcal{F}_{K}$ is taken as the loss function. {To find the global minimum of the loss function $\mathcal{F}_{K}$, we update the parameters $\bm\theta$ till the termination condition is reached. We denote the obtained optimal parameters by $\bm\theta_{opt}$. Then we can prepare an approximation operator for the Gibbs state by performing the parameterized circuit $U(\bm\theta_{opt})$}.

{The loss function $\mathcal{F}_{K}$ is obtained by truncating the Taylor series of von Neumann entropy at order $K$. Specifically, let $K\in\mathbb{Z}_{+}$ be a positive integer, and denote the truncated entropy by $S_K(\rho)$}. Let $H$ denote the Hamiltonian and $\beta>0$ be the inverse temperature, then the loss function $\mathcal{F}_{K}(\bm\theta)$ is defined as 
\begin{align}
\mathcal{F}_{K}(\bm\theta)=\tr(H\rho(\bm\theta))-\beta^{-1}S_{K}(\rho(\bm\theta)).
\label{EQ:40}
\end{align} 
Here the free energy is determined by parameters $\bm\theta$, since the state $\rho(\bm\theta)$ is prepared by the PQC $U(\bm\theta)$. Particularly, the $K$-truncated entropy $S_{K}(\rho)$ is given as follows,
\begin{align}
S_{K}(\rho)=\sum_{k=1}^{K}\frac{(-1)^{k}}{k}\tr\left((\rho-I)^{k}\rho\right)=\sum_{j=0}^{K}C_{j}\tr(\rho^{j+1}). \label{EQ:12}
\end{align}
In Eq.~\eqref{EQ:12}, coefficients $C_{j}$'s of $S_{K}(\rho)$ are given in the form below:
\begin{align}
C_{0}=\sum_{k=1}^{K}\frac{1}{k},\,
C_{j}=\sum_{k=j}^{K}\binom{k}{j}\frac{(-1)^{j}}{k}\,,
C_{K}=\frac{(-1)^{K}}{K}.    \label{EQ:13}
\end{align}
where $j=1,...,K-1$. 

Recall that our goal is to find parameters $\bm\theta_{opt}$ that minimize the value of the loss function $\mathcal{F}_{K}(\bm\theta)$, i.e., $\bm\theta_{\bm opt}=\text{argmin}_{\bm\theta} \mathcal{F}_{K}(\bm\theta)$. In practice, the optimization loop only terminates if some condition given previously is reached. Therefore, one cannot obtain the true optimum, but some parameters $\bm\theta_{0}$ that will approximately minimize the loss function in the sense that
\begin{equation}
\mathcal{F}_{K}(\bm\theta_{0})\leq\min_{\bm\theta}\mathcal{F}_{K}(\bm\theta)+\epsilon,
\label{EQ:5}
\end{equation}
where $\epsilon$ is the error tolerance in the optimization problem. {Especially, we assume the used PQC $U(\bm\theta)$ endows sufficient expressiveness to prepare the desired Gibbs state or a state very close to it. Hence, the state $\rho(\bm\theta_0)$ could be taken to approximate the desired Gibbs state.}

\paragraph{Error analysis}
Since the loss function $\mathcal{F}_{K}(\bm\theta)$ is a truncated version of the free energy, the solution $\bm\theta_{0}$ to the optimization problem in Eq.~\eqref{EQ:5} {is not exactly the quantum Gibbs state $\rho_{G}$}. However, the obtained state $\rho(\bm\theta_{0})$ is not far away from the Gibbs state $\rho_{G}$. {Here, we use the fidelity to characterize the distance between two states.} In the following, we show the validity of this claim by providing a lower bound on the fidelity between $\rho(\bm\theta_{0})$ and $\rho_{G}$ in Theorem~\ref{TH:2}. {In particular, the result in Theorem~\ref{TH:2} implies that the larger truncation order $K$ is, the state $\rho(\bm\theta_0)$ is closer to the state $\rho_G$.}


\begin{theorem}
\label{TH:2}
Given a positive integer $K$ and error tolerance $\epsilon>0$, let $\beta>0$ be the inverse temperature, and $\bm\theta_{0}$ be the solution to the optimization in Eq.~\eqref{EQ:5}. Assume the rank of the output state $\rho(\bm\theta_{0})$ is $r$, then the fidelity between the state $\rho(\bm\theta_{0})$ and the Gibbs state $\rho_{G}$ is lower bounded as follows
{
\begin{equation}
F(\rho(\bm\theta_{0}),\rho_{G})\geq 1-\sqrt{2{\left(\beta\epsilon+\frac{2r}{K+1}\left(1-\Delta\right)^{K+1}\right)}},
\label{EQ:fidelity}
\end{equation}}
where $\Delta\in(0,e^{-1})$ is a constant determined by $K$. 
\end{theorem} 

Theorem~\ref{TH:2} implies that we can regard the output state $\rho(\bm\theta_{0})$ as an approximation for the Gibbs state for a given error tolerance in the optimization problem and a truncation order $K$. And the quantity in the right-hand-side of Eq.~\eqref{EQ:fidelity} quantifies the extent that $\rho(\bm\theta_{0})$ approximates $\rho_{G}$. {Also, we can easily see that the quantity becomes larger when the order $K$ increases.}


Next, we prove Theorem~\ref{TH:2} by building a connection between the relative entropy and the fidelity and then deriving an upper bound on the truncation error. 
\begin{lemma}
Given quantum states $\rho$ and $\sigma$ and a constant $\delta>0$, suppose that the relative entropy $S(\rho\|\sigma)$ is less than $\delta$, i.e., $S(\rho\|\sigma)\leq\delta$. Then the fidelity between $\rho$ and $\sigma$ is lower bounded. To be specific, $F(\rho,\sigma)\geq 1-\sqrt{2\delta}$.
\label{PR:2}
\end{lemma}
\begin{proof}
Recall the relationship between the trace distance and the fidelity $D(\rho,\sigma)\geq 1-F(\rho,\sigma)$, and Pinsker's inequality $D(\rho,\sigma)\leq \sqrt{2S(\rho\|\sigma)}$, then we have the following inequality,
\begin{equation}
F(\rho,\sigma)\geq 1-D(\rho,\sigma)\geq 1-\sqrt{2S(\rho\|\sigma)}\geq 1-\sqrt{2\delta}.
\end{equation}
\end{proof}

Lemma~\ref{PR:2} states that, if one wants to lower bound the fidelity $F(\rho(\bm\theta_{0}),\rho_{G})$ between the obtained state $\rho(\bm\theta_{0})$ and the Gibbs state $\rho_{G}$, then it suffices to upper bound the relative entropy $S(\rho(\bm\theta_{0})\|\rho_{G})$ between them. Thus we proceed to give an upper bound of the relative entropy. 

Let $\delta_{0}$ be the truncation error of $S_{K}(\rho)$, then the definition of the free energy allows to bound the difference between the free energy and its truncated version, i.e., {$|\mathcal{F}_{K}(\rho)-\mathcal{F}(\rho)|\leq\beta^{-1}\delta_{0}$}. Recall the well-known free energy equation, $\mathcal{F}(\rho)=\mathcal{F}(\rho_{G})+\beta^{-1}S(\rho\|\rho_{G})$, which indicates that, for arbitrary density $\rho$, the free energy $\mathcal{F}(\rho)$ can be represented as a linear combination of the free energy $\mathcal{F}(\rho_{G})$ of the quantum Gibbs state $\rho_G$ and the relative entropy between $\rho$ and $\rho_{G}$. Therefore, an upper bound on the relative entropy $S(\rho(\bm\theta_{0})\|\rho_{G})$ is readily derived as follows:
{
\begin{align}
S(\rho(\bm\theta_{0})\|\rho_{G})=&\beta|\mathcal{F}(\rho(\bm\theta_{0}))-\mathcal{F}(\rho_{G})|\\
=&\beta|\mathcal{F}(\rho(\bm\theta_{0}))-\mathcal{F}_{K}(\rho(\bm\theta_{0}))\nonumber\\
&+\mathcal{F}_{K}(\rho(\bm\theta_{0}))-\mathcal{F}(\rho_{G})| \\
=&\delta_{0}+\beta|\mathcal{F}_{K}(\rho(\bm\theta_{0}))-\mathcal{F}(\rho_{G})| \\
\leq& 2\delta_{0}+\beta\epsilon,\label{EQ:6}
\end{align}
where the inequality in Eq.~\eqref{EQ:6} is due to the fact that $\mathcal{F}(\rho_{G})\leq\mathcal{F}_{K}(\rho(\bm\theta_{0}))\leq\mathcal{F}(\rho_{G})+\beta^{-1}\delta_{0}+\epsilon$, which is stated in Lemma~\ref{le:loss_function_inequality} and proved in Appendix~\ref{supp:le:loss_function_inequality}. In particular, to obtain the result in Lemma~\ref{le:loss_function_inequality}, we assume the used PQC is expressive enough to prepare the target Gibbs state or a state very close to it.}

\begin{lemma}\label{le:loss_function_inequality}
Given the error tolerance $\epsilon>0$ in the optimization problem in Eq.~\eqref{EQ:5}, suppose the truncation error of the free energy is $\beta^{-1}\delta_{0}>0$. Then we can derive a relation between $\mathcal{F}(\rho_{G})$ and $\mathcal{F}_{K}(\rho(\bm\theta_{0}))$ below, where $\bm\theta_{0}$ is the output of the optimization and $\rho_{G}$ is the Gibbs state.
\begin{align}
    \mathcal{F}(\rho_{G})\leq\mathcal{F}_{K}(\bm\theta_{0})\leq\mathcal{F}(\rho_{G})+\beta^{-1}\delta_{0}+\epsilon.
    \label{eq:proof_loss_function}
\end{align}
\end{lemma}

{Now, given the truncation order $K$, we derive an upper bound on the difference between $S_K(\rho)$ and $S(\rho)$ in the following lemma and defer the proof to Appendix~\ref{sec:lemma3}.}
\begin{lemma}
\label{PR:3}
Given a quantum state $\rho$, assume the truncation order of the truncated von Neumann entropy is $K\in\mathbb{Z}_{+}$, and choose $\Delta\in(0,e^{-1})$ such that $-\Delta\ln(\Delta)<\frac{1}{K+1}(1-\Delta)^{K+1}$. Let $\delta_{0}$ denote the truncation error, i.e., the difference between the von Neumann entropy $S(\rho)$ and its $K$-truncated entropy $S_K(\rho)$. Then the truncated error $\delta_{0}$ is upper bounded in the sense that 
\begin{align}
\delta_{0}\leq\frac{r}{K+1}\left(1-\Delta\right)^{K+1},
\label{EQ:38}
\end{align}
where $r$ denotes the rank of density operator.
\end{lemma}

Replacing $\delta_{0}$ in Eq.~\eqref{EQ:6} with its upper bound in Eq.~\eqref{EQ:38} immediately leads to a bound on the relative entropy $S(\rho(\bm\theta_{0})\|\rho_{G})$. Taking this bound into Lemma~\ref{PR:2}, a lower bound on the fidelity $F(\rho(\bm\theta),\rho_{G})$ is then derived, which is exactly the one in Eq.~\eqref{EQ:fidelity}. Now, the proof of Theorem~\ref{TH:2} is completed.

\subsection{Optimization via the gradient-based method}
Finding optimal parameters $\bm\theta_{opt}$ is a major part of our variational algorithm. Both gradient-based and gradient-free methods could be used to do the optimization. Here, we provide analytical details on the gradient-based approach, and we refer to \cite{Benedetti2019a} for more information on the optimization subroutines in variational quantum algorithms.

{The choice of truncation order $K$ could be various and depends on the required accuracy for Gibbs state preparation. Here we mainly focus on the two-order loss function $\mathcal{F}_{2}(\bm\theta)$ as higher fidelity could be expected by the result in Theorem~\ref{TH:2}. We numerically show the validity of this choice in the next section. Particularly, the numerical results show we can use a two-order loss function to prepare high-fidelity Gibbs states of several many-body Hamiltonians.}

Now, we show that $\mathcal{F}_{2}(\bm\theta)$ is convex, which indicates that the gradient-based method could efficiently minimize it. We also derive the analytical expressions for its gradients and show that these analytical gradients could also be evaluated efficiently on NISQ devices. {Especially, the same circuit for estimating $\mathcal{F}_{2}(\rho)$ can also be used to calculate their gradients.}

\paragraph{Convexity of $2$-truncated free energy}
Recall the definition of $K$-truncated entropy $S_{K}(\rho)$ in Eq.~\eqref{EQ:12}, and, in this section, we take $K=2$. Given truncation order $2$, the loss function $\mathcal{F}_{2}(\rho)$ is defined in the following form:
\begin{equation}
\mathcal{F}_{2}(\bm\theta)=\tr(H\rho(\bm\theta))+\beta^{-1}\left(2\tr(\rho(\bm\theta)^{2})-\frac{1}{2}\tr(\rho(\bm\theta)^{3})-\frac{3}{2}\right).
\label{EQ:35}
\end{equation}

Notice that the functional $\tr(H\rho)$ is linear for a given Hamiltonian $H$ and $\beta>0$, therefor the convexity of loss function $\mathcal{F}_{2}$ is determined by the convexity of the functional $g(\rho)=2\tr(\rho^{2})-\frac{1}{2}\tr(\rho^{3})-\frac{3}{2}$. Hence, to prove the convexity of $\mathcal{F}_{2}$, we only need to show the convexity of functional $g(\rho)$.

\begin{lemma}
\label{PR:1}
The functional $g(\rho)=2\tr(\rho^{2})-\frac{1}{2}\tr(\rho^{3})-\frac{3}{2}$ is convex, where $\rho$ is a density operator.
\end{lemma}
\begin{proof}
According to Theorem 2.10 of \cite{Carlen2010a}, the functional $\tr(f(\rho))$ is convex if the associated function $f:\mathbb{R}\to\mathbb{R}$ is convex. In the scenario, where $\mathcal{F}_{2}(\rho)$ is given in Eq.~\eqref{EQ:35}, the associated function of $g$ is defined as $f(x)=2x^{2}-\frac{1}{2}x^{3}-\frac{3}{2}$ for all $x\in[0,1]$. The claim follows from proving that $f$ is convex, and the second order derivative $f^{(2)}$ of $f$ is positive, since 
\begin{equation}
f^{(2)}(x)=4-3x\geq 1\quad \forall x\in[0,1].
\end{equation}
Therefore, the positivity of the second order derivative of $f$ leads to the convexity of $\mathcal{F}_{2}(\rho)$ in the set of densities operators. 
\end{proof}

\paragraph{Analytical gradient}
Here we discuss the computation of the gradient of the global loss function $\mathcal{F}_{2}(\bm\theta)$. Inspired by previous works~\cite{mitarai2018quantum,Schuld2019,Ostaszewski2019}, we compute the gradients of the 2-truncated free energy $\mathcal F_2$ by shifting the parameters of the same circuit for estimating $\mathcal F_2$. Note that there is an alternative method to estimate the partial derivative with a single circuit~\cite{Farhi2018}, but at the cost of using an ancillary qubit.

In Fig.~\ref{fig:Gibbs prepare}, the density operator $\rho(\bm\theta)$ is prepared in the register $B$ by performing a sequence of unitaries $U=U_{N}...U_{1}$ on the state $\ket{00}_{AB}$ in registers $AB$. Each gate $U_{m}$ is either fixed, e.g., a C-NOT gate, or parameterized. The parameterized gates are of the form $U_{m}=e^{-iH_{m}\theta_{m}/2}$, where $\theta_{m}$'s are real parameters and $H_{m}$'s are a tensor product of Pauli matrices. Then the loss function $\mathcal{F}_{2}$ is related to parameters $\bm\theta$, and its gradient is explicitly given by 
\begin{align}
\nabla_{\bm\theta}\mathcal{F}_{2}(\bm\theta)=\left(\frac{\partial\mathcal{F}_{2}(\bm\theta)}{\partial\theta_{1}},...,\frac{\partial\mathcal{F}_{2}(\bm\theta)}{\partial\theta_{N}}\right).
\label{EQ:41}
\end{align}
Particularly, the partial derivatives of $\mathcal{F}_2$ with respect to $\theta_m$ is 
\begin{align}
\frac{\partial\mathcal{F}_{2}(\bm\theta)}{\partial\theta_{m}}=&\frac{1}{2}({\langle K\rangle_{\theta_{m}+\frac{\pi}{2}}-\langle K\rangle_{\theta_{m}-\frac{\pi}{2}}})\nonumber\\
&+\beta^{-1}[(2(\langle O\rangle_{\theta_{m}+\frac{\pi}{2},\theta_{m}}-\langle O\rangle_{\theta_{m}-\frac{\pi}{2},\theta_{m}})\nonumber\\
&-\frac{3}{4}(\langle G\rangle_{\theta_{m}+\frac{\pi}{2},\theta_{m},\theta_{m}}-\langle G\rangle_{\theta_{m}-\frac{\pi}{2},\theta_{m},\theta_{m}})], \label{EQ:37}
\end{align}
where $\langle K\rangle$, $\langle O\rangle$, and $\langle G\rangle$ are used to estimate $\tr(H\rho(\bm\theta))$, $\tr(\rho(\bm\theta)^2)$, and $\tr(\rho(\bm\theta)^3)$, respectively, and their definitions 
are given as follows:
\begin{align}
&\langle K\rangle_{\theta_{\alpha}}=\tr(U_{\alpha}\psi_{A_{1}B_{1}} U_{\alpha}^{\dagger}\cdot K),\label{eq:notation:k}\\
&\langle O\rangle_{\theta_{\alpha},\theta_{\beta}}=\tr([\otimes_{l=2}^{3}U_{\alpha}\psi_{A_{l}B_{l}} U_{\alpha}^{\dagger}]\cdot O),\label{eq:notation:o}\\
&\langle G\rangle_{\theta_{\alpha},\theta_{\beta},\theta_{\gamma}}=\tr([\otimes_{l=4}^{6}U_{\alpha}\psi_{A_{l}B_{l}} U_{\alpha}^{\dagger}]\cdot G).\label{eq:notation:g}
\end{align}
{Here, we defer the definitions of notations $K,O,G$ and the process of deriving the gradient to the Appendix~\ref{app:gradient}.} From Eq.~\eqref{EQ:37}, we can see that the gradient can be efficiently computed by shifting the parameters in the loss function.

\begin{figure}[htb]
	\centering
		\subfigure[\ Ansatz\ with 6 parameters ]{\label{fig:Ising_ansatz_6_paras}
		\includegraphics[width=0.35\textwidth]{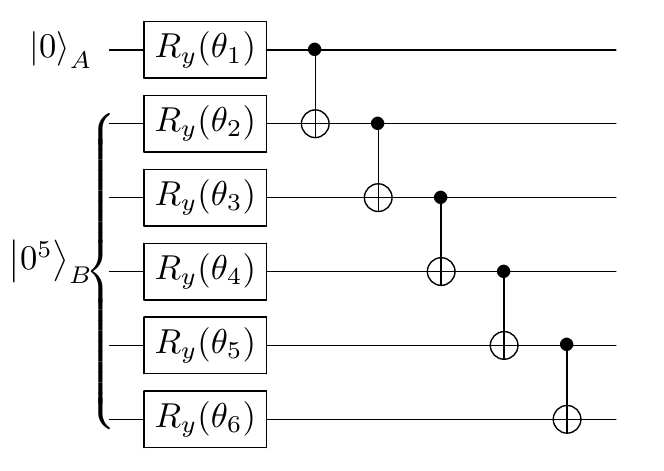}}
\subfigure[\ Ansatz\ with only 1 parameter]{\label{fig:Ising_ansatz_1_para}
		\includegraphics[width=0.35\textwidth]{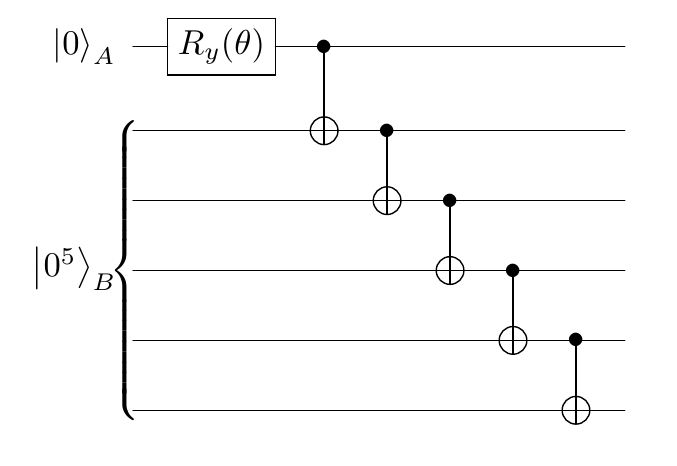}}
	\caption{Two ansatzes for Ising chain model. These ansatzes are composed of two registers $A$ and $B$, where one ancillary qubit is set in $A$ and $5$ qubits are set in $B$. Notably, the qubits in $B$ are performed with rotations $R_{y}(\theta)$ and CNOT gates in (a), while only CNOT gates in (b).}
	\label{fig:Ising_ansatz}
\end{figure}

To summarize, the above results entail that the partial derivatives of our loss function $\mathcal{F}_{2}(\bm\theta)$ with respect to $\bm\theta$ are completely determined by Eq.~\eqref{EQ:37}. This indicates that the analytical gradient of our loss function $\mathcal{F}_{2}(\bm\theta)$ can be efficiently computed on near-term quantum devices by shifting parameters and performing measurements.
With the analytical gradients, one could apply the gradient-based methods to minimize the loss function. Specifically, parameters $\bm\theta$ in the unitary $U(\bm\theta)$ are updated towards the steepest direction of the loss function, i.e.,
\begin{align}
\bm\theta\leftarrow \bm\theta-\eta\nabla_{\bm\theta}\mathcal{F}_{2}(\bm\theta),
\end{align}
where $\nabla_{\bm\theta}\mathcal{F}_{2}(\bm\theta)$ is the gradient vector and $\eta$ is the learning rate that determines the step sizes. Under suitable assumptions, the loss functions converge to the global minimum after certain iterations of the above procedure. Notice that other gradient-based methods, e.g., stochastic gradient descent,  ADAM~\cite{kingma2014adam}, can also be used in the optimization loop of our variational Gibbs state preparation algorithm.

\begin{figure}[htb]
	\centering
		\subfigure[\ Ansatz\ with 6 parameters ]{\label{fig:ising6}
		\includegraphics[width=0.45\textwidth]{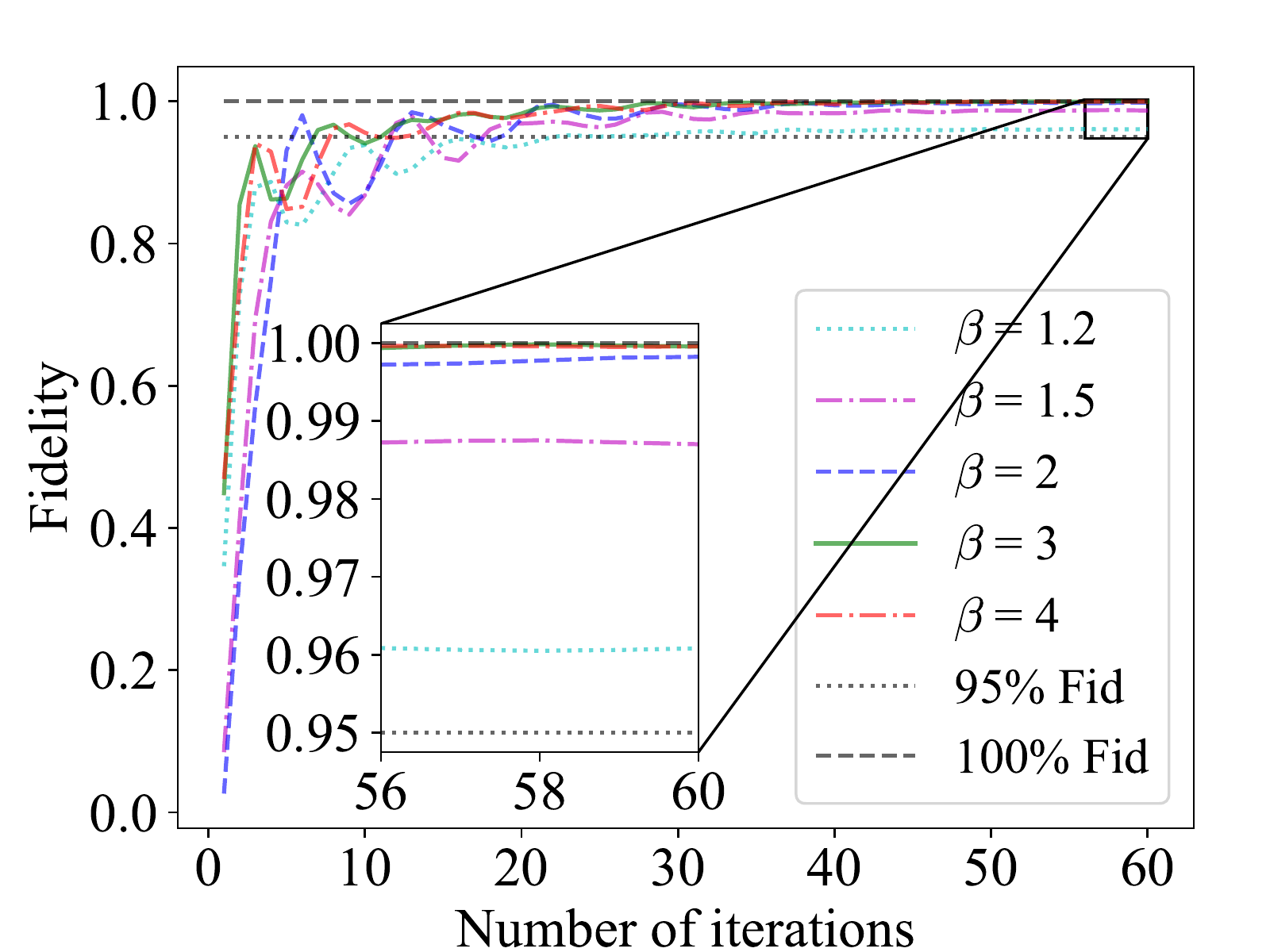}}
\subfigure[\ Ansatz\ with only 1 parameter]{\label{fig:ising1}
		\includegraphics[width=0.45\textwidth]{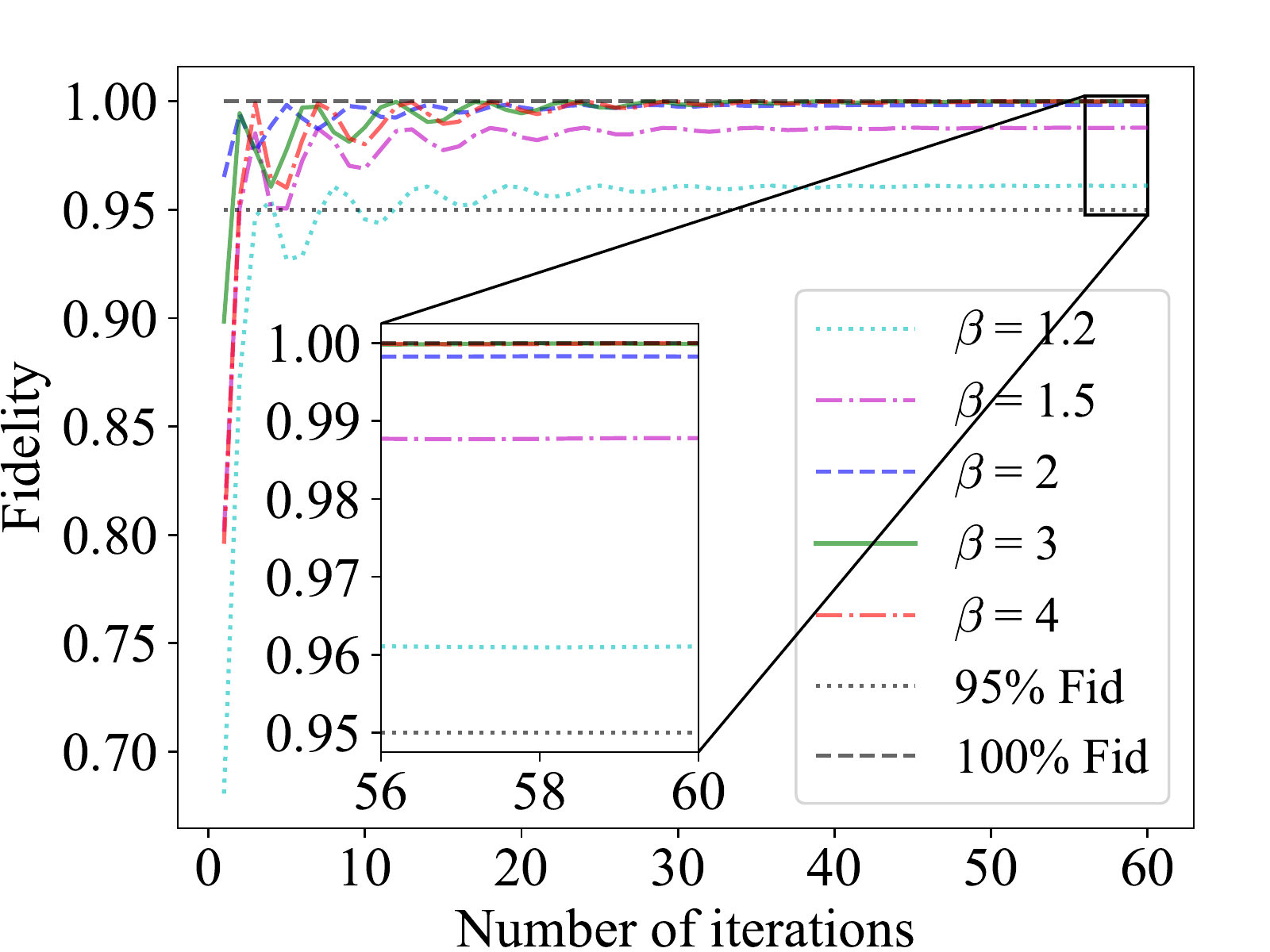}}
	\caption{Fidelity curves for the Ising chain Gibbs state preparation with different $\beta$. In (a), we use the Ansatz with $6$ parameters (cf. Fig.~\ref{fig:Ising_ansatz_6_paras}); In (b), we use the Ansatz with only 1 parameter (cf. Fig.~\ref{fig:Ising_ansatz_1_para}).
	We can see that they have almost the same performance, which indicates only 1 parameter is enough for this task.}
	\label{fig:Ising_fidelity}
\end{figure}

\begin{figure}[h]
	\centering
		\includegraphics[width=0.45\textwidth]{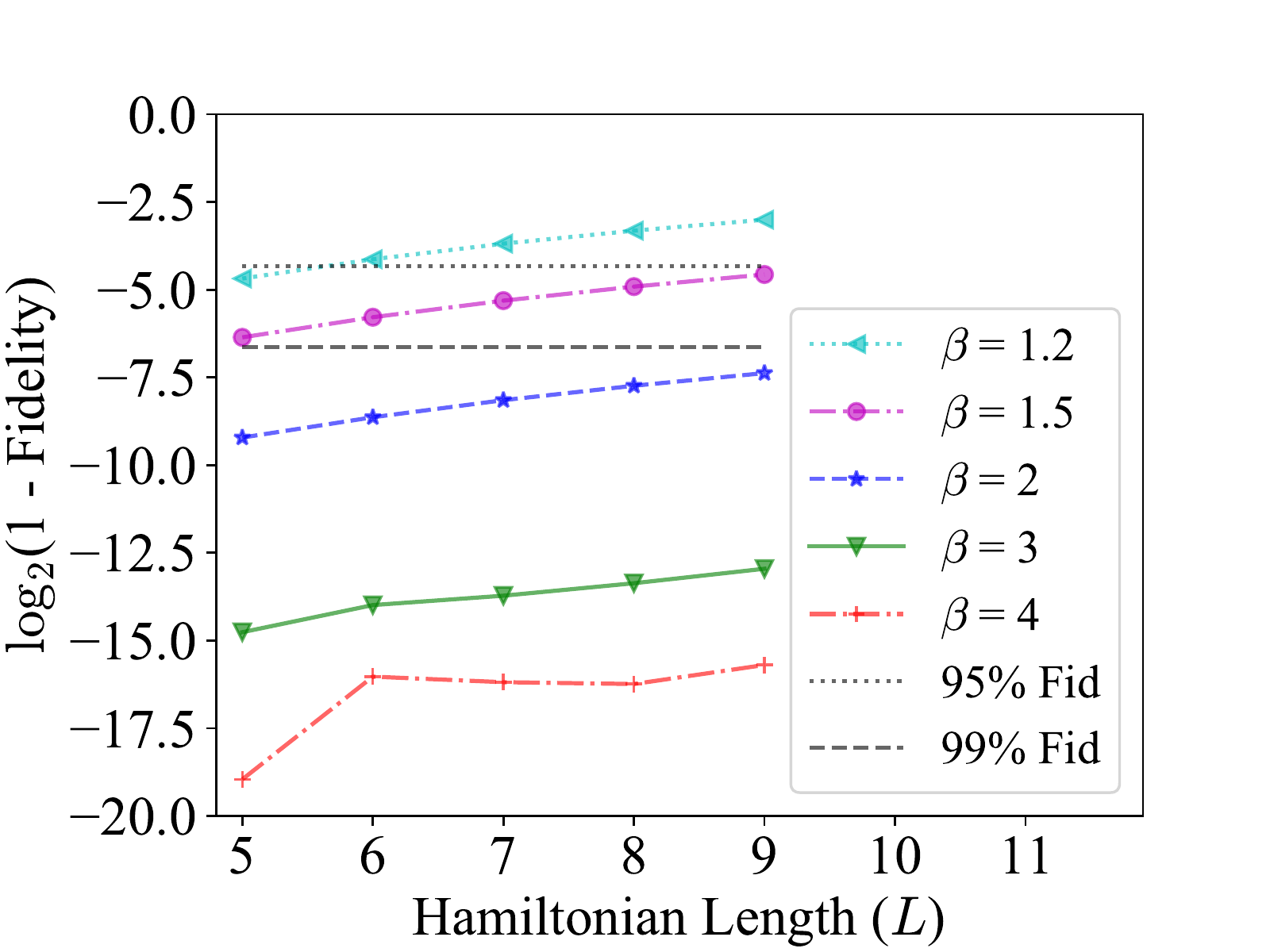}
	\caption{Semilog plot of the fidelity vs. the Ising Hamiltonian length ($L$) with different $\beta$ for the Ising chain model. Here, $\log_2$ means logarithm with base $2$. We can see that the fidelity increases exponentially with $\beta$ growing. }
	\label{fig:log_fidelity}
\end{figure}

\section{Numerical simulations}
\label{SEC:EXPERIMENT}
This section has conducted numerical experiments for preparing the Gibbs states of two common Hamiltonian examples: the Ising chain model and XY spin-1/2 chain model. Specifically, in subsection~\ref{sec:Ising}, we study the Ising chain model and show that a parameterized circuit with one ancillary qubit and shallow depth could be trained to produce the Ising Gibbs state with fidelity higher than 99\%, especially at higher $\beta$ or lower inverse temperature ($\beta\ge 2$). Furthermore, we also give a more sophisticated ansatz with only one parameter that can do the same thing. As the second example, in subsection~\ref{sec:spin chain}, we study the spin chain model. We show that our approach could achieve a fidelity higher than 95\% via an ansatz with 30 parameters for $\beta\ge 1.5$ for a 5-length XY spin-1/2 chain Hamiltonian. In particular, the fidelity could also achieve 99\% for the lower inverse temperature case.
In our numerical experiments, the classical parameters of the parameterized circuits are initialized from a uniform distribution in $[0, 2\pi]$, and then updated via the ADAM gradient-based method~\cite{kingma2014adam} until the loss function is converged.

\subsection{Ising model}\label{sec:Ising}
As our first example, we consider the
spin 1/2 chain $B$ of length $L=5$, with the Ising Hamiltonian 
\begin{align}
H_{B}=-\sum_{i=1}^{L} Z_{B, i} Z_{B, i+1}
\label{ising_model}
\end{align}
and periodic boundary conditions (i.e., $Z_{B,6}=Z_{B,1}$). Our goal is to prepare the corresponding Gibbs state 
\begin{align}\label{eq:Ising_Gibbs}
\rho_G = e^{-\beta H_{B}}/\tr(e^{-\beta H_{B}}).
\end{align}

To prepare this state, we choose a $6$-qubit parameterized circuit with $n_A=1$ and $n_B=5$ (cf. Fig. \ref{fig:Ising_ansatz}), where $A$ is the ancillary system that used for producing a mixed state on system $B$. Here we need to note that we only use a 1-qubit ancillary system in our ansatz, which is significantly less than \cite{Wu2019b} where $n_A=n_B$.
In Fig.~\ref{fig:Ising_ansatz_6_paras}, the ansatz consists of 6 single qubit Pauli-Y rotation operators with different classic parameters $\theta_i$ ($i\in [6]$) and 5 CNOT gates.

After applying this ansatz on the input zero state $\ket{0}_A\ket{0^5}_B$, we can get the output state on system B, which is desired to get close to the Gibbs state in Eq. \eqref{eq:Ising_Gibbs}. The fidelity between this output state and the Gibbs state, in the training process with different $\beta$, is depicted in Fig. \ref{fig:ising6}. When $\beta \ge 1.2$, after $30$ iterations of updating parameters, our method can easily achieve a fidelity higher than 95\%. Specifically, if $\beta\ge 2$, the fidelity is higher than 99\%, which indicates that our approach can almost exactly prepare the Gibbs state in Eq.~\eqref{eq:Ising_Gibbs}, especially at higher inverse temperatures.

We also test the preparation of the Ising Gibbs state for different length (i.e., $L=5,6,7,8,9$), and all of the ansatzes are similar to Fig. \ref{fig:Ising_ansatz_6_paras}, which only uses one additional qubit. The curves of the logarithmic form of the fidelity between the output state $\rho_B$ and the Gibbs state $\rho_G$ are depicted in Fig. \ref{fig:log_fidelity}. We can intuitively see that the larger the Hamiltonian length is, the lower fidelity we achieve. However, we can also find that the temperature has a significant impact on fidelity: the larger the $\beta$ is, the higher the fidelity is (see Proposition \ref{prop:Fid_lower_bound} for a detailed analysis). In particular, when $\beta\ge 2$, the fidelity is already higher than 99\%, for the Hamiltonian length we have listed in the figure.

An interesting experimental phenomenon in the training process is that the first parameter $\theta_1$ in the system $A$ is approaching $\pi/2$ while other parameters in the system $B$ are approaching $0$. Hence, we update the ansatz to a simplified one in Fig.~\ref{fig:Ising_ansatz_1_para} and implement the numerical simulations in Fig.~\ref{fig:ising1}. Notably, the overall performance is almost the same as using the ansatz with $6$ parameters (cf.~Fig.\ref{fig:ising6}). To further explore this interesting behaviour of the Ising chain Gibbs state preparation, we analyze the states generated using different loss functions.

\begin{proposition}\label{prop:circuit_with_1_para}
Given the circuit in Fig.~\ref{fig:Ising_ansatz_1_para} and denote $\rho_B(\theta)$ its output state on system B. For the Ising chain model, if we compute its free energy in Eq.~\eqref{eq:free_energy} and our truncated cost in Eq.~\eqref{EQ:40}, then the optimal parameters that minimize these two loss functions are both
\begin{align}
    \theta= \pi/2 + k \pi,
\end{align}
where $k\in \mathbb{Z}$. As a result, $\rho_B(\pi/2)$ is the best state, under this circuit, that approaches the Gibbs state in Eq. \eqref{eq:Ising_Gibbs}, with a fidelity larger than 95\% for any $\beta \ge 1.25$.
\end{proposition}
We defer the proof of Proposition\ref{prop:circuit_with_1_para} to Appendix~\ref{sec:prop:circuit_with_1_para}. Here, we need to note that this fidelity is just a lower bound, actually, when $\beta=1.2$, we have still achieved a fidelity greater than 95\% for $n_B=5$, as demonstrated in our experiments. And in Proposition \ref{prop:circuit_with_1_para}, the number of qubits in system $B$ is not limited to 5, instead it can be any positive integer that greater than two.

Another interesting experimental result (c.f. Fig. \ref{fig:log_fidelity}) shows that the fidelity between the Ising chain Gibbs state $\rho_G$ and the output state $\rho_{B}$ of our method increases exponentially when $\beta$ increases. The result is as follows, and the details can be found in Appendix~\ref{sec:proof:Fid_lower_bound}.
\begin{proposition}\label{prop:Fid_lower_bound}
Given the circuit in Fig.~\ref{fig:Ising_ansatz_1_para} and let $\rho_{B}(\theta)$ be its output state on system $B$. Then the fidelity between $\rho_{B}(\pi/2)$ and the Gibbs state $\rho_{G}$ is lower bounded. To be more specific,
\begin{align}
F(\rho_{B}(\pi/2),\rho_{G})\geq\frac{1}{\sqrt{1+(N/2-1)e^{-\beta\Delta}}},
\label{EQ:fidelity:beta}
\end{align}
where $N$ is the dimension of system $B$, i.e., $N=2^{n_B}$ and {$\Delta$ is the spectral gap of the Hamiltonian $H_B$ on system $B$, i.e., the discrepancy between the minimum and the second minimum eigenvalues.}
\end{proposition}

Notably, our approach prepare Gibbs state with high accuracy when $\beta$ is large, as the fidelity $F(\rho_{B}(\pi/2),\rho_{G})$ converges fast to 1 with $\beta$ increasing.
\begin{figure*}[h]
	\centering
		\includegraphics[width=0.8\textwidth]{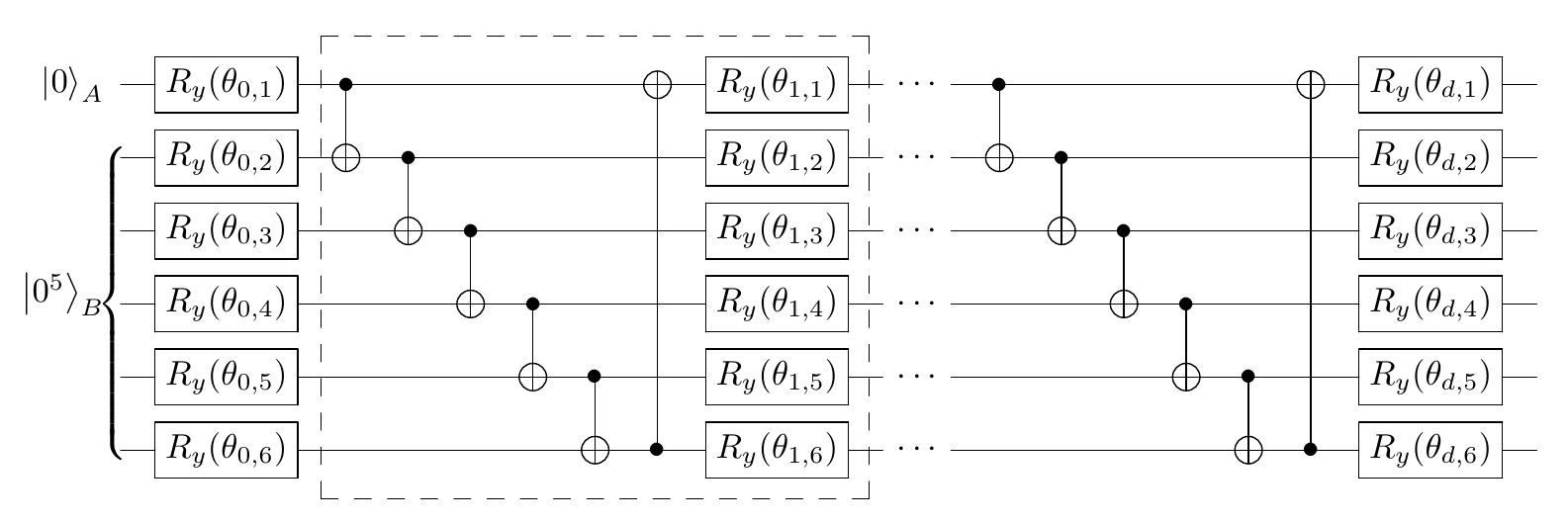}
	\caption{The ansatz for XY spin-1/2 chain model. In this ansatz, it contains one ancilla qubit in register $A$ and $5$ qubits in register $B$. Rotation gates $R_{y}(\theta)$ are first applied on all qubits. Then, a basic circuit module (denoted in the dashed-line box) composed of CNOT gates and rotation gates $R_{y}(\theta)$ is repeatedly applied. Here, $d$ means repeating $d$ times.}
	\label{fig:XYspin_ansatz}
\end{figure*}

 \begin{figure*}[h]
	\centering
	\subfigure[$\ d=3$]{\label{fig:xyspin3}
		\includegraphics[width=0.45\textwidth]{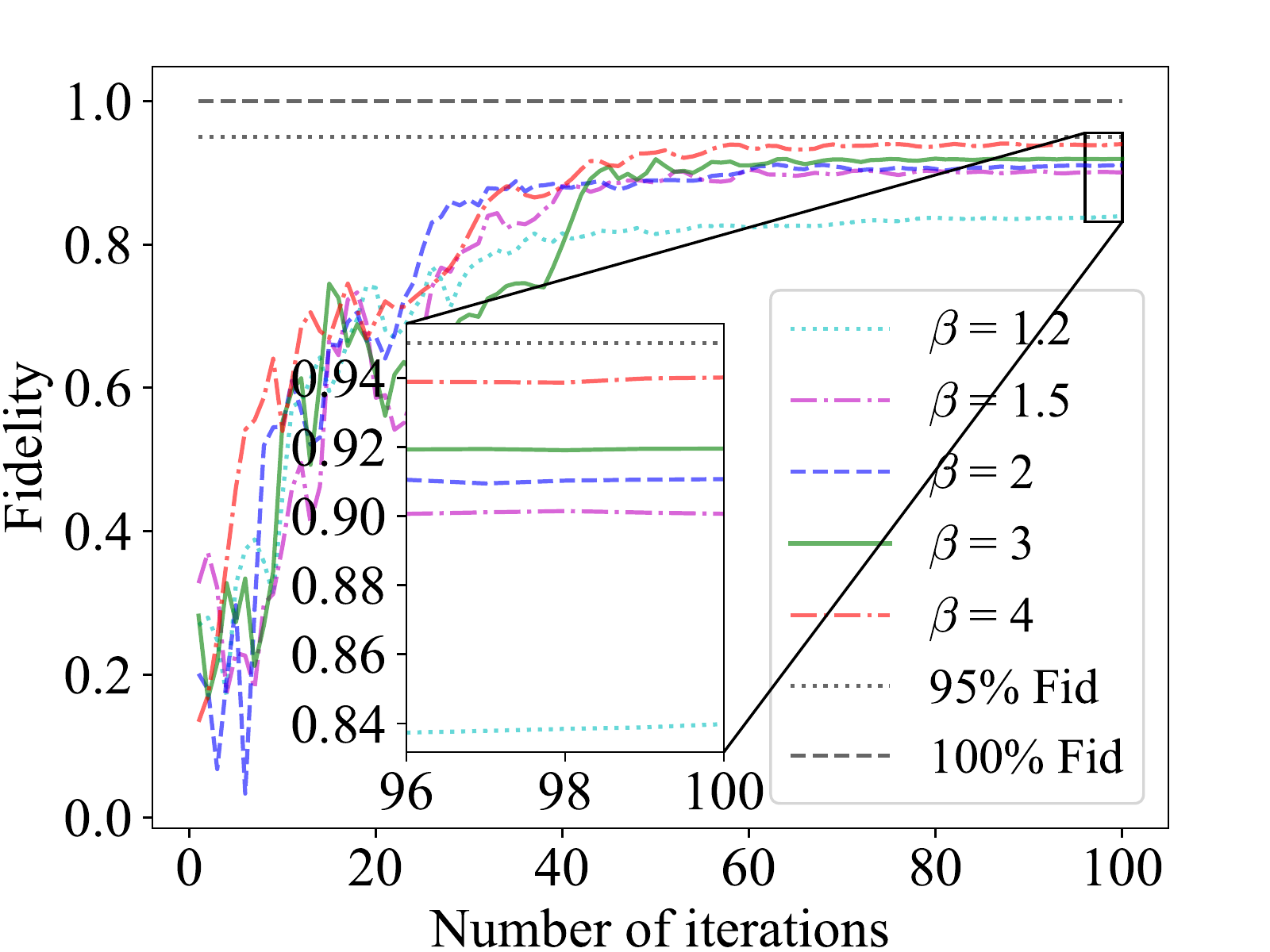}}
	\subfigure[$\ d=4$]{\label{fig:xyspin4}
\includegraphics[width=0.45\textwidth]{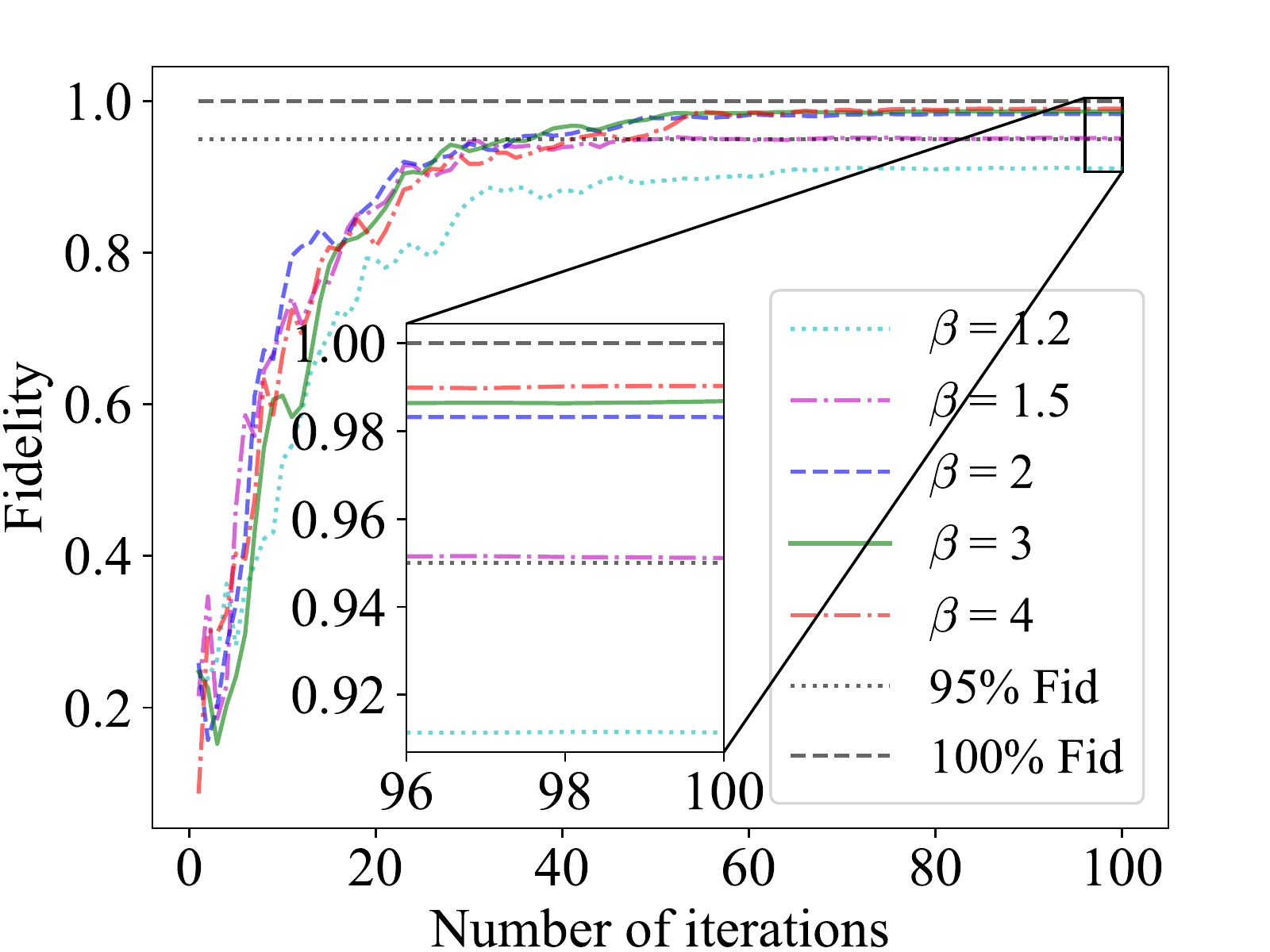}}
	\subfigure[$\ d=5$]{\label{fig:xyspin5}
\includegraphics[width=0.45\textwidth]{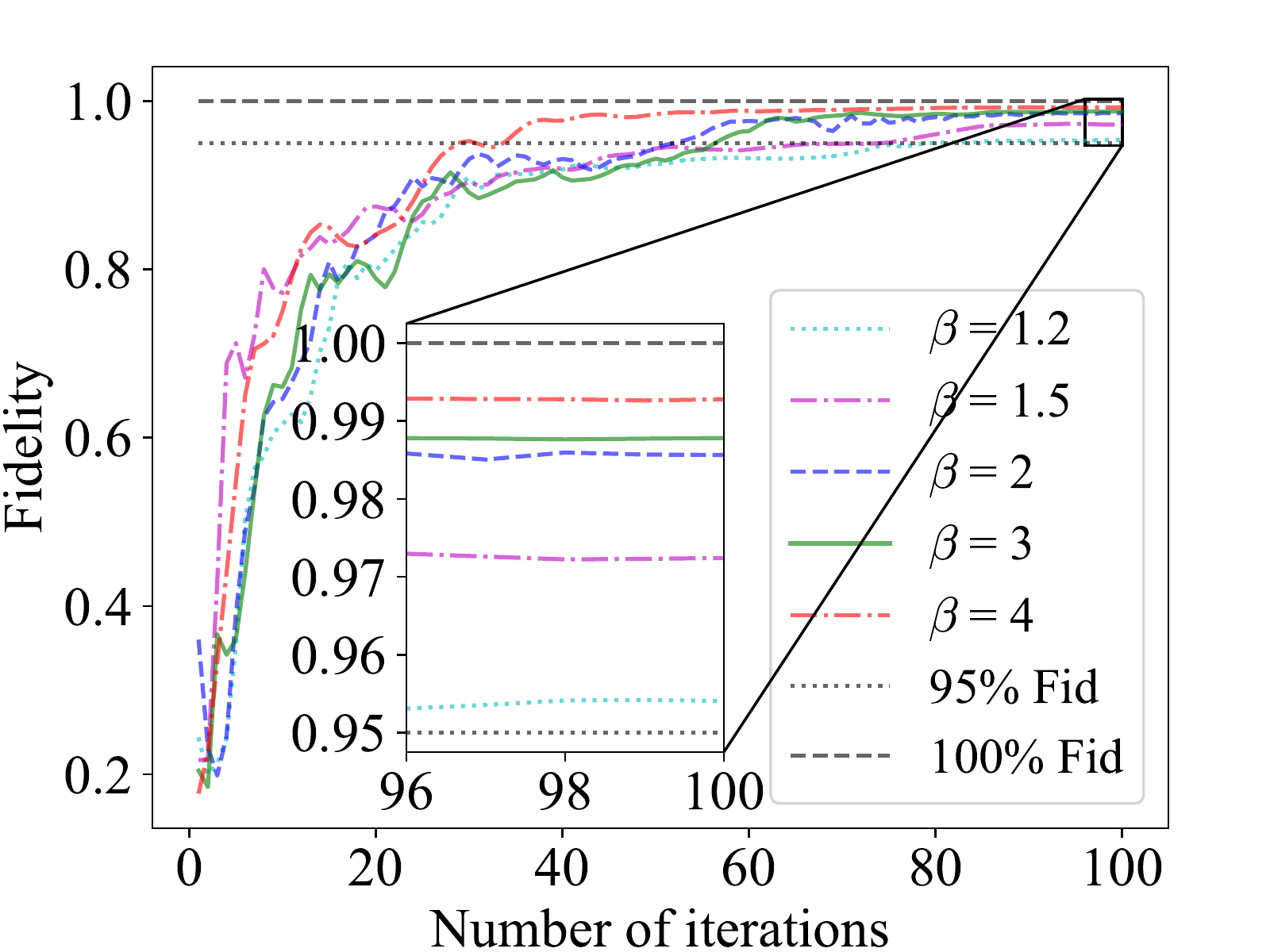}}
	\subfigure[$\ d=6$]{\label{fig:xyspin6}
\includegraphics[width=0.45\textwidth]{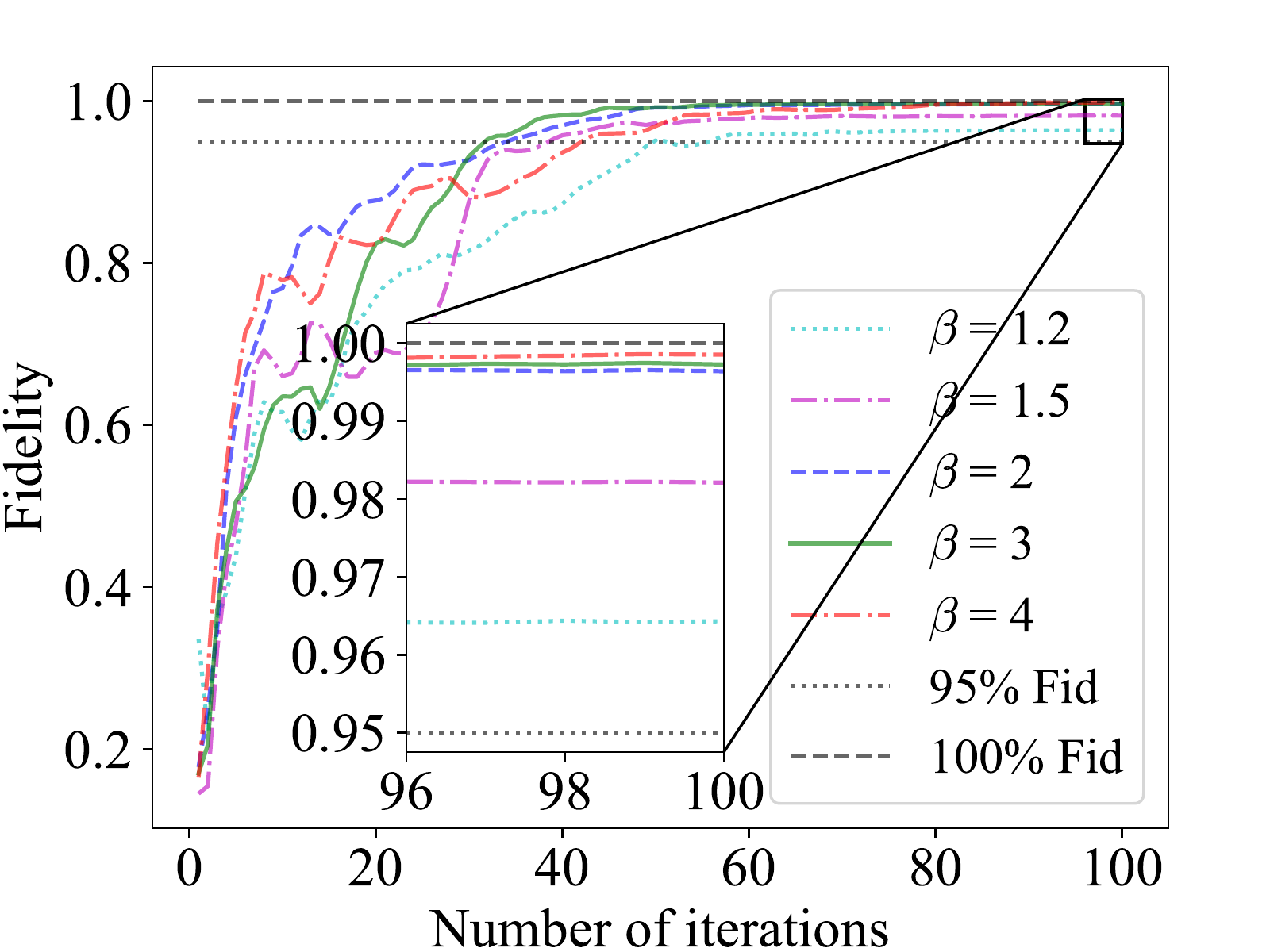}}
	\caption{Fidelity curves for the XY spin-1/2 chain Gibbs state preparation with different $\beta$. The results of the fidelity obtained with different $\beta$ are represented by coloured lines. In (a)-(d), numerical experiments are performed using different ansatzes. In each ansatz, the basic circuit module (cf. Fig.~\ref{fig:XYspin_ansatz}) is repeated different times, i.e., $d$. Note that each ansatz has $(n_A+n_B)(d+1) = 6(d+1)$ parameters. Here better performance are obtained with larger $d$.}
	\label{fig:xyspin}
\end{figure*}

\subsection{XY spin-1/2 chain model}
 \label{sec:spin chain}
Our second instance is the XY spin-1/2 chain $B$ of length $L=5$, with the Hamiltonian
 \begin{align}\label{eq:XYspin_H}
     H_B=-\sum_{i=1}^{L} X_{B,i}X_{B,i+1}+Y_{B,i}Y_{B,i+1}
 \end{align}
 and periodic boundary conditions (i.e., $Z_{B,6}=Z_{B,1}$).
 
To prepare the spin chain Gibbs state, we first choose a $6$-qubit parameterized circuit with $n_A=1$ and $n_B=5$ (cf. Fig.~\ref{fig:XYspin_ansatz}), where the basic circuit module (which contains a CNOT layer and a layer of single-qubit Pauli-Y rotation operators) is repeated $d$ times, and the total number of parameters of this circuit is $(n_A+n_B)(d+1)$.

The fidelity between the output state of this circuit and the Gibbs state is shown in Fig.~\ref{fig:xyspin}, where different $d$'s are included. We see that when $d\geq 4$ and $\beta\ge 1.5$, our approach can easily achieve a fidelity greater than 95\% and if $\beta\ge 2$, the fidelity could be higher than 98\%. Furthermore, if $\beta$ is equaling to 4, the fidelity can be even higher than 99\%, which means our approach could almost generate the Gibbs state exactly in higher $\beta$ (or lower temperature). One possible reason that we need larger $d$ for this instance is that the Hamiltonian in Eq.~\eqref{eq:XYspin_H} is not directly generated via the CNOT module. Hence we will need multiple CNOT modules to fully entangle the state.

{
We need to note that the above experiments mainly focus on lower temperatures, i.e., $\beta>1$, where smaller ancillary systems and lower truncation order $K=2$ are usually sufficient to achieve a higher fidelity.
In order to test our algorithm's performance with different truncation order $K$ under higher temperatures (e.g., $\beta<0.5$), we choose a $6$-qubit parameterized circuit with $n_A=3$ and $n_B=3$ for a 3-length XY spin-1/2 chain Hamiltonian. Here, the ansatz is similar to Fig.~\ref{fig:XYspin_ansatz} and we set $d=8$ to make it expressive enough. $n_A=n_B$ means the ancillary systems can cover all the Hamiltonian's ranks. The boxplots of the fidelity versus various truncation order $K$ are illustrated in Fig.~\ref{fig:boxplot_fidelity_order_K}, where three higher temperatures ($\beta=0.1,0.2,0.3$) are included. From the results, we could intuitively see that the larger truncation order, the higher fidelity we achieve, which is in line with our theoretical analysis. And the phenomenon that we could still achieve a fidelity higher than 98\% even under higher temperatures indicates our hybrid algorithm has a powerful ability in preparing Gibbs states of certain many-body Hamiltonians.
}

\begin{figure}[htb]
	\centering
\includegraphics[width=0.45\textwidth]{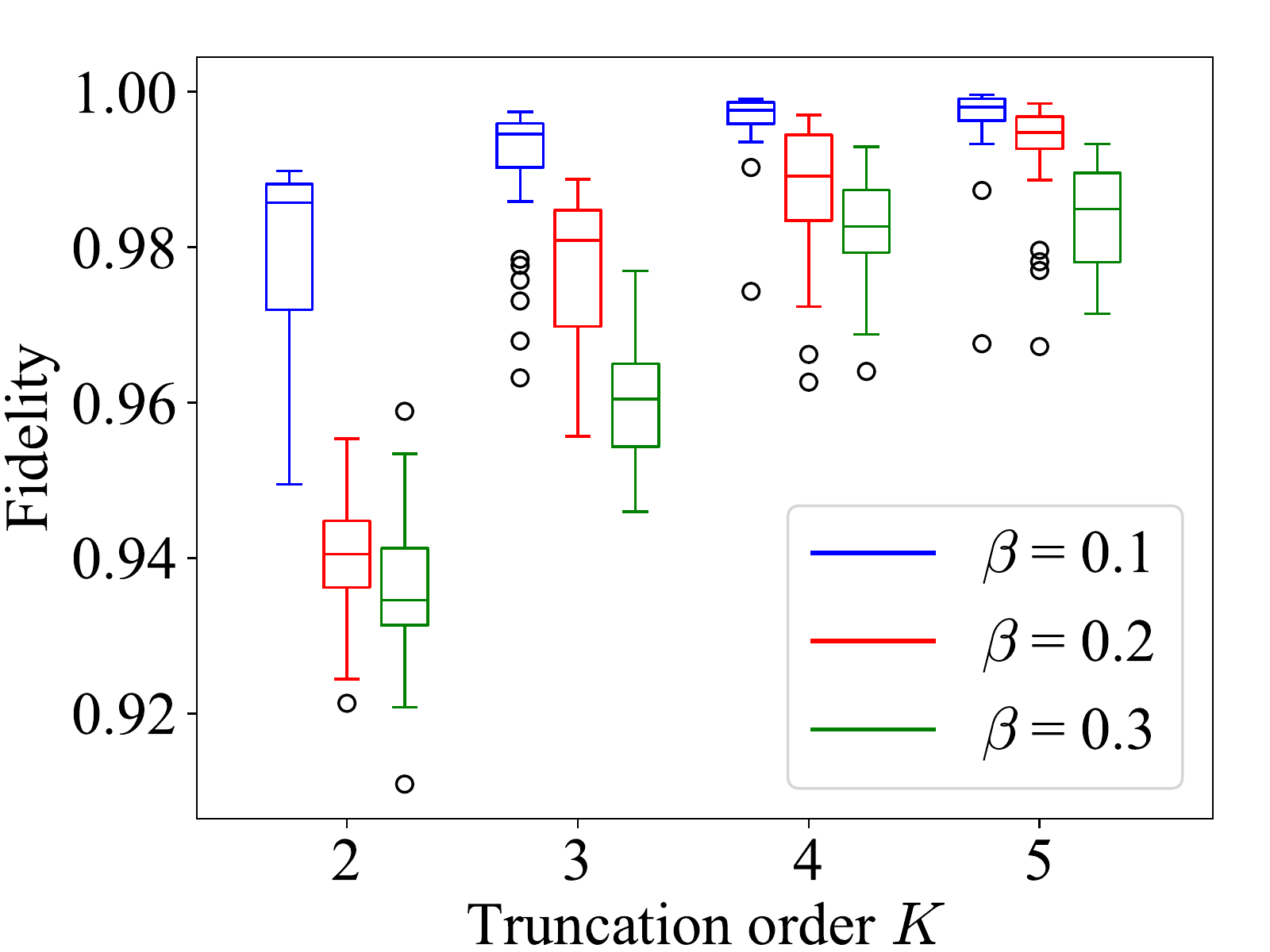}
	\caption{Boxplot of the fidelity vs. the truncation order $K$ with different $\beta$ for the XY spin-1/2 chain model. Here the ansatz is similar to Fig.~\ref{fig:XYspin_ansatz} while $n_A=n_B=3$. Each box consists of 30 runs with different parameter initializations.}
	\label{fig:boxplot_fidelity_order_K}
\end{figure}


\section{Discussions}
\label{SEC:DISCUSSION}
In this work, we provide hardware-efficient variational algorithms for quantum Gibbs state preparation with NISQ devices.
We design loss functions to approximate the free energy of a given hamiltonian by utilizing the truncated Taylor series of the von Neumann entropy. By minimizing the loss functions, the parameterized quantum circuits can be trained to learn the Gibbs state via variational algorithms
since the Gibbs state minimizes free energy. In particular, we show that both the loss functions and their gradients can be evaluated on NISQ devices, thus allowing us to implement the hybrid quantum-classical optimization via either gradient-based or gradient-free optimization methods. Moreover, we show that our method could efficiently prepare the Gibbs states via analytical evidence and numerical experiments.

We further show that our variational algorithms work efficiently for many-body models, including the Ising chain and spin chain. In particular, we show that the preparation of the Ising Gibbs state can be done efficiently and accurately via shallow parameterized quantum circuits with only one parameter and one additional qubit. We expect that our results may shed light on quantum optimization, quantum simulation, and quantum machine learning in the NISQ era.  

\section*{Acknowledgements}
We would like to thank Runyao Duan, Yuan Feng, and Sanjiang Li for helpful discussions.
Y. W. and G. L. contribute equally to this work and acknowledge support from the Baidu-UTS AI Meets Quantum project.  G. L. acknowledges the financial support from China Scholarship Council (No. 201806070139). This work was partly supported by the Australian Research Council (Grant No: DP180100691).

%



\onecolumngrid

\begin{center}
{\textbf{\large Appendix  }}
\end{center}

\renewcommand{\theequation}{S\arabic{equation}}
\renewcommand{\thealgorithm}{S\arabic{algorithm}}
\setcounter{equation}{0}
\setcounter{figure}{0}
\setcounter{table}{0}
\setcounter{section}{0}
\setcounter{proposition}{0}
\setcounter{lemma}{1}
\setcounter{algorithm}{0}

\section{Variational algorithm for Gibbs state preparation with higher-order truncations}
\label{appendix:detals}
Here we present a variational algorithm for preparing the Gibbs state with $K$-truncated free energy. To illustrate our algorithm, we give some notations first. We let $A_{t_{j}}B_{t_{j}}$ denote the registers that store the states for estimating $\tr(\rho^{t})$, where $t_j\in\Sigma_t$ and $\Sigma_t$ includes all the indices of these registers.

\begin{figure}[h]
\begin{algorithm}[H] 
\caption{Variational quantum Gibbs state preparation with truncation order $K$}
\begin{algorithmic}[1] \label{alg:rftl:K}
\STATE Input: choose the ansatz of unitary $U(\bm\theta)$, tolerance $\varepsilon$, truncation order $K$, and initial parameters of $\bm\theta$;

\STATE Compute coefficients $C_{0}$, $C_{1}$,..., $C_{K}$ according to Eq.~\eqref{EQ:13}. 

\STATE Prepare initial states $\ket{00}$ in registers $AB$ and apply  $U(\bm\theta)$ to these states.%
 
\STATE Measure and compute  $\tr (H\rho_{B_1})$ and compute the loss function $L_1 =  \tr (H\rho_{B_1} )$;

\STATE Measure the overlap $\tr (\prod_{t_2\in\Sigma_2}\rho_{B_{t_2}})$ via Destructive Swap Test and compute the loss function $L_2 =-\beta^{-1}C_{1}\tr(\prod_{t_2\in\Sigma_2}\rho_{B_{t_2}})$; 

\STATE Measure the overlap $\tr(\prod_{t_{k}\in\Sigma_{k}}\rho_{B_{t_{k}}})$ via higher order state overlap estimation and compute the loss function $L_{k}=-\beta^{-1}C_{k-1}\tr(\prod_{t_{k}\in\Sigma_{k}}\rho_{B_{t_{k}}})$ for each $k\in\{3,...,K+1\}$.

\STATE Perform optimization of $\mathcal{F}_K(\bm\theta) = \sum_{k=1}^{K+1}L_k-\beta^{-1}C_{0}$ and update parameters of $\bm\theta$;

\STATE Repeat 3-7 until the loss function $\mathcal F_K(\bm\theta)$ converges with tolerance $\varepsilon$;
\STATE Output the state $\rho^{out} = \tr_A U(\bm\theta)\op{00}{00}_{AB}U(\bm\theta)^\dagger$.
\end{algorithmic}
\end{algorithm}
\end{figure}

\section{Technical details}
\subsection{Supplementary proof for Lemma~\ref{le:loss_function_inequality}}\label{supp:le:loss_function_inequality}
\begin{lemma}
Given the error tolerance $\epsilon>0$ in the optimization problem in Eq.~\eqref{EQ:5}, suppose the truncation error of the free energy is $\beta^{-1}\delta_{0}>0$. Then we can derive a relation between $\mathcal{F}(\rho_{G})$ and $\mathcal{F}_{K}(\rho(\bm\theta_{0}))$ below, where $\bm\theta_{0}$ is the output of the optimization and $\rho_{G}$ is the Gibbs state.
\begin{align}
    \mathcal{F}(\rho_{G})\leq\mathcal{F}_{K}(\bm\theta_{0})\leq\mathcal{F}(\rho_{G})+\beta^{-1}\delta_{0}+\epsilon.
    \label{eq:proof_loss_function}
\end{align}
\end{lemma}

\begin{proof}
First, we show that left inequality in Eq.~\eqref{eq:proof_loss_function}. For arbitrary density operator $\rho$, we have $\mathcal{F}_{K}(\rho)-\mathcal{F}(\rho)>0$. To be specific,
\begin{align}
\mathcal{F}_{K}(\rho)-\mathcal{F}(\rho)&=\beta^{-1}(S(\rho)-S_{K}(\rho))\\
&=-\beta^{-1}\tr\left(\sum_{j=K+1}^{\infty}\frac{(-1)^{j+1}}{j}(\rho-I)^{j}\rho\right)>0 \label{eq:proof:entropy}
\end{align}
In Eq.~\eqref{eq:proof:entropy}, we expand the von Neumann entropy into the Taylor series, i.e., $S(\rho)=-\tr\left(\sum_{j=1}^{\infty}\frac{(-1)^{j+1}}{j}(\rho-I)^{j}\rho\right)$, and the result holds immediately. 

Second, the right inequality in Eq.~\eqref{eq:proof_loss_function} is a direct result of the definition of truncated free energy $\mathcal{F}_{K}(\rho)$.
\begin{align}
    \mathcal{F}_{K}(\bm\theta_{0})-\mathcal{F}(\rho_{G})&=\mathcal{F}_{K}(\bm\theta_{0})-\min_{\bm\theta}\mathcal{F}_{K}(\bm\theta)+\min_{\bm\theta}\mathcal{F}_{K}(\bm\theta)-\mathcal{F}(\rho_{G})\label{eq:proof:entropy_4}\\
    &\leq \epsilon+\mathcal{F}_{K}(\rho_{G})-\mathcal{F}(\rho_{G})\label{eq:proof:entropy_5}\\
    &\leq \epsilon+\beta^{-1}\delta_{0}, \label{eq:proof:entropy_6}
\end{align}
where we use the fact that $\min_{\bm\theta}\mathcal{F}_{K}(\bm\theta)\leq \mathcal{F}_{K}(\bm\theta_{0})\leq\min_{\bm\theta}\mathcal{F}_{K}(\bm\theta)+\epsilon$ in Eq.~\eqref{eq:proof:entropy_5}, and the inequality in Eq.~\eqref{eq:proof:entropy_6} is due to the truncation. {Especially, we here assume the PQC endows sufficient expressiveness to prepare the desired Gibbs state or a state very close to it, which allows $\min_{\bm\theta}\mathcal{F}_{K}(\bm\theta)\leq\mathcal{F}_{K}(\rho_G)$.}
\end{proof}

\subsection{Supplementary proof for Lemma~\ref{PR:3}}\label{sec:lemma3}
\begin{lemma}
Given a quantum state $\rho$, assume the truncation order of the truncated von Neumann entropy is $K\in\mathbb{Z}_{+}$, and choose $\Delta\in(0,e^{-1})$ such that $-\Delta\ln(\Delta)<\frac{1}{K+1}(1-\Delta)^{K+1}$. Let $\delta_{0}$ denote the truncation error, i.e., the difference between the von Neumann entropy $S(\rho)$ and its $K$-truncated entropy $S_K(\rho)$. Then the truncated error $\delta_{0}$ is upper bounded in the sense that 
\begin{align}
\delta_{0}\leq\frac{r}{K+1}\left(1-\Delta\right)^{K+1},
\end{align}
where $r$ denotes the rank of density operator.
\end{lemma}

\begin{proof}
The proof proceeds by expanding the logarithm function in the entropy into Taylor series. The upper bound of the difference between the entropy $S(\rho)$ and its truncated version $S_{K}(\rho)$ for density $\rho$ is given as follows,
\begin{align}
    \delta_{0}&=\left|S(\rho)-S_{K}(\rho)\right|\\
    &=\left|\tr\left(\sum_{k=K+1}^{\infty}\frac{(-1)^{k}}{k}(\rho-I)^{k}\rho\right)\right|\\
    &=\left(\sum_{j:~\lambda_{j}\geq\Delta}+\sum_{j:~0<\lambda_{j}<\Delta}\right)\sum_{k=K+1}^{\infty}\frac{\lambda_{j}}{k}(1-\lambda_{j})^{k}. 
    \label{EQ:spectral}
\end{align}
In the above Eq.~\eqref{EQ:spectral} we use the spectral decomposition of $\rho=\sum_{j=1}^{r}\lambda_{j}\op{\psi_{j}}{\psi_{j}}$. 

To give an upper bound on truncation error $\delta_{0}$, we give upper bounds on two terms in Eq.~\eqref{EQ:spectral}. First, we consider the term with eigenvalues larger than $\Delta$.
\begin{align}
    &\sum_{j:~\lambda_{j}\geq\Delta}\sum_{k=K+1}^{\infty}\frac{\lambda_{j}}{k}(1-\lambda_{j})^{k}\nonumber\\
    &=\sum_{j:~\lambda_{j}\geq\Delta}\sum_{k=K+1}^{\infty}\left[\frac{1}{k}(1-\lambda_{j})^{k}-\frac{1}{k}(1-\lambda_{j})^{k+1}\right]\label{eq:error_lambda}\\
    &\leq\sum_{j:~\lambda_{j}\geq\Delta}\sum_{k=K+1}^{\infty}\frac{1}{k}(1-\lambda_{j})^{k}-\sum_{j:~\lambda_{j}\geq\Delta}\sum_{k=K+1}^{\infty}\frac{1}{k+1}(1-\lambda_{j})^{k+1} \label{eq:errur_deflation}\\
    &=\frac{1}{K+1}\sum_{j:~\lambda_{j}\geq\Delta}(1-\lambda_{j})^{K+1}.\label{EQ:39}
\end{align}
The equality in Eq.~\eqref{eq:error_lambda} is due to the substitution of $\lambda_{j}$ with $1-(1-\lambda_{j})$, and the inequality in~\eqref{eq:errur_deflation} follows by replacing $1/k$ with $1/(k+1)$ in the right summation of Eq.~\eqref{eq:error_lambda}.
 
Then we consider the term with non-zero eigenvalues less than $\Delta$.
\begin{align}
    \sum_{j:~0<\lambda_{j}<\Delta}\lambda_{j}\sum_{k=K+1}^{\infty}\frac{1}{k}(1-\lambda_{j})^{k}&\leq\sum_{j:~0<\lambda_{j}<\Delta}-\lambda_{j}\ln(\lambda_{j}) \label{eq:error:entropy}\\
    &\leq \sum_{j:~0<\lambda_{j}<\Delta}-\Delta\ln(\Delta),\label{eq:error:delta}
\end{align}
where the inequality in Eq.~\eqref{eq:error:entropy} follows from replacing the series with $-\ln(\lambda_{j})$, since function $S(x)=-x\ln(x)=\sum_{l=1}^{\infty}\frac{1}{l}x(1-x)^{l}$, and the second inequality is due to the fact that $S(x)$ increases as $x$ increases in the interval $(0,e^{-1})$.

In all, an upper bound on $\delta_{0}$ can be given as
\begin{align}
    \delta_{0}&\leq \frac{1}{K+1}\sum_{j:~\lambda_{j}\geq\Delta}(1-\lambda_{j})^{K+1}+\sum_{j:~0<\lambda_{j}<\Delta}-\Delta\ln(\Delta)\\
    &\leq r\cdot\left(\frac{r_{0}}{r}\frac{(1-\Delta)^{K+1}}{K+1}+\frac{r_{1}}{r}(-\Delta\ln(\Delta))\right)\label{eq:error:combination}\\
    &\leq r\cdot\max\{\frac{(1-\Delta)^{K+1}}{K+1},-\Delta\ln(\Delta)\}.\label{eq:error:max}
\end{align}
where $r_{0}$ ($r_{1}$) denotes the number of non-zero eigenvalues larger (less) than $\Delta$. As $-\Delta\ln(\Delta)<\frac{1}{K+1}(1-\Delta)^{K+1}$, the claim is proved.
\end{proof}

\subsection{Estimation of the higher-order gradients}
\label{sec:Supplemental_higher_order_gradient}
\renewcommand{\thelemma}{S\arabic{lemma}}
\setcounter{lemma}{0}
\begin{lemma}\label{lemma:optimization:1}
Given a parameterized density operator $\rho(\bm\theta)$, we have the following equality,
\begin{align}
    \partial_{\theta_{m}}\tr(\rho(\bm\theta)^{3})=3\partial_{\theta_{m},1}\tr(\rho_{1}(\bm\theta)\otimes\rho_{2}(\bm\theta)\otimes\rho_{3}(\bm\theta)\cdot S_{1}S_{2}),
\end{align}
where $\partial_{\theta_{m},1}$ means the derivative is computed with respective to $\theta_{m}$ of the state stored in $1$-th register, $\rho_{j}(\bm\theta)$ is the state stored in $j$-th register, and $S_{1}={\rm SWAP}_{12}\otimes I_{3}$ and $S_{2}=I_{1}\otimes{\rm SWAP}_{23}$, and the ${\rm SWAP}_{ij}$ is the operator that swaps the state stored in $i$-th and $j$-th register.
\end{lemma}
\begin{proof}
To prove the claim, we need the following result, which we give the proof later. 
\begin{align}
\tr(\rho_{1}\otimes\rho_{2}\otimes\rho_{3}\cdot S_{1}S_{2})=\overline{\tr(\rho_{2}\otimes\rho_{1}\otimes\rho_{3}\cdot S_{1}S_{2})}=\overline{\tr(\rho_{3}\otimes\rho_{2}\otimes\rho_{1}\cdot S_{1}S_{2})}.
\label{eq:optimization:appendix:1}
\end{align}
Let $\rho_{1}(\bm\theta)=\rho_{2}(\bm\theta)=\rho_{3}(\bm\theta)=\rho(\bm\theta)$, then the claimed is proved in the following,
\begin{align}
\frac{\partial}{\partial \theta_{m}}\tr(\rho(\bm\theta)^{3})&=\frac{\partial}{\partial \theta_{m}}\tr(\rho_{1}(\bm\theta)\otimes\rho_{2}(\bm\theta)\otimes\rho_{3}(\bm\theta)\cdot S_{1}S_{2})\\
&=\frac{\partial}{\partial \theta_{m,1}}\tr(\rho_{1}(\bm\theta)\otimes\rho_{2}(\bm\theta)\otimes\rho_{3}(\bm\theta)\cdot S_{1}S_{2}) \nonumber\\
&+\frac{\partial}{\partial \theta_{m,2}}\tr(\rho_{1}(\bm\theta)\otimes\rho_{2}(\bm\theta)\otimes\rho_{3}(\bm\theta)\cdot S_{1}S_{2}) \nonumber\\
&+\frac{\partial}{\partial \theta_{m,3}}\tr(\rho_{1}(\bm\theta)\otimes\rho_{2}(\bm\theta)\otimes\rho_{3}(\bm\theta)\cdot S_{1}S_{2})\label{eq:optimization:appendix:2}\\
&=\frac{\partial}{\partial \theta_{m,1}}\tr(\rho_{1}(\bm\theta)\otimes\rho_{2}(\bm\theta)\otimes\rho_{3}(\bm\theta)\cdot S_{1}S_{2})\nonumber\\
&+\frac{\partial}{\partial \theta_{m,2}}\tr(\rho_{2}(\bm\theta)\otimes\rho_{1}(\bm\theta)\otimes\rho_{3}(\bm\theta)\cdot S_{1}S_{2}) \nonumber\\
&+\frac{\partial}{\partial \theta_{m,3}}\tr(\rho_{3}(\bm\theta)\otimes\rho_{2}(\bm\theta)\otimes\rho_{1}(\bm\theta)\cdot S_{1}S_{2})\label{eq:optimization:appendix:3}\\
&=3\frac{\partial}{\partial \theta_{m,1}}\tr(\rho_{1}(\bm\theta)\otimes\rho_{2}(\bm\theta)\otimes\rho_{3}(\bm\theta)\cdot S_{1}S_{2}),\label{eq:optimization:appendix:4}
\end{align}
where the equality is Eq.~\eqref{eq:optimization:appendix:2} is the result of chain rule, and we use the relation in Eq.~\eqref{eq:optimization:appendix:1} to derive the equality in Eq.~\eqref{eq:optimization:appendix:3}.

Now we prove the equalities in Eq.~\eqref{eq:optimization:appendix:1}. 

Let $\rho_{1}=\sum_{j}p_{j}\op{\phi_{j}}{\phi_{j}}$, $\rho_{2}=\sum_{l}q_{l}\op{\psi_{l}}{\psi_{l}}$, and $\rho_{3}=\sum_{k}r_{k}\op{\xi_{k}}{\xi_{k}}$. We have the following equalities. 
\begin{align}
    \tr(\rho_{1}\otimes\rho_{2}\otimes\rho_{3}\cdot S_{1}S_{2})&=\sum_{jlk}p_{j}q_{l}r_{k}\ip{\psi_{l}}{\phi_{j}}\ip{\xi_{k}}{\psi_{l}}\ip{\phi_{j}}{\xi_{k}},\label{eq:optimization:5}\\
    \tr(\rho_{2}\otimes\rho_{1}\otimes\rho_{3}\cdot S_{1}S_{2})&=\sum_{jlk}p_{j}q_{l}r_{k}\ip{\phi_{j}}{\psi_{l}}\ip{\xi_{k}}{\phi_{j}}\ip{\psi_{l}}{\xi_{k}},\label{eq:optimization:6}\\
    \tr(\rho_{3}\otimes\rho_{2}\otimes\rho_{1}\cdot S_{1}S_{2})&=\sum_{jlk}p_{j}q_{l}r_{k}\ip{\psi_{l}}{\xi_{k}}\ip{\phi_{j}}{\psi_{l}}\ip{\xi_{k}}{\phi_{j}}.\label{eq:optimization:7}
\end{align}
Comparing Eqs.~\eqref{eq:optimization:5}-\eqref{eq:optimization:7}, the equality in Eq.~\eqref{eq:optimization:appendix:1} is proved.
\end{proof}

\subsection{Sumplementary discussion for optimization}\label{app:gradient}
To simplify the notations, let $L_{1}$ denote $\tr(H\rho(\bm\theta))$, $L_{2}$ denote $2\beta^{-1}\tr(\rho(\bm\theta)^{2})$, and $L_{3}$ denote $-\frac{\beta^{-1}}{2}\tr(\rho(\bm\theta)^{3})$. Using these notations, our loss function can be rewritten as $\mathcal{F}_{2}(\bm\theta)=L_{1}+L_{2}+L_{3}-\frac{3\beta^{-1}}{2}$, and the gradient of $\mathcal{F}_{2}(\bm\theta)$ can be rewritten as follow:
\begin{align}
\nabla_{\bm\theta}\mathcal{F}_{2}(\bm\theta)=\nabla_{\bm\theta}L_{1}+\nabla_{\bm\theta}L_{2}+\nabla_{\bm\theta}L_{3}.
\end{align} 
Therefore, the gradient of $\mathcal{F}_{K}(\bm\theta)$ can be estimated via computing the gradients of $L_{j}$, $j=1,2,3$. Specifically, $\nabla_{\bm\theta}L_{j}$, $j=2,3$, can be computed using the destructive SWAP test and higher order state overlap estimation, respectively. As for $\nabla_{\bm\theta} L_{1}$, it can be estimated by measurement directly.

Next, we show that the gradients of $L_{j}$'s can be computed by shifting the parameters $\bm\theta$ of the circuit. 
The partial derivatives of each $L_{j}$ have the following forms, {
\begin{align}
\frac{\partial L_{1}}{\partial \theta_{m}}&=
\frac{\partial}{\partial\theta_{m}}\tr(U_{N}...U_{1}\op{0}{0}U_{1}^{\dagger}...U_{N}^{\dagger}\cdot(I\otimes H)),\label{EQ:42}\\%
\frac{\partial L_{2}}{\partial\theta_{m}}&=
2\beta^{-1}\frac{\partial}{\partial\theta_{m}}\tr((U_{N}...U_{1}\op{0}{0}U_{1}^{\dagger}...U_{N}^{\dagger})^{\otimes 2}\cdot W_{1}),\label{EQ:44} \\%
\frac{\partial L_{3}}{\partial\theta_{m}}&=-\frac{\beta^{-1}}{2}
\frac{\partial}{\partial\theta_{m}}\tr((U_{N}...U_{1}\op{0}{0}U_{1}^{\dagger}...U_{N}^{\dagger})^{\otimes3}\cdot W_{2}),\label{EQ:46}
\end{align}}
where $W_{1}$ denotes the operator ${\rm SWAP}_{B_{2}B_{3}}\otimes I_{A_{2}A_{3}}$, $W_{2}$ denotes the operator $({\rm SWAP}_{B_{4}B_{5}}\otimes I_{A_{4}A_{5}A_{6}B_{6}})\cdot({\rm SWAP}_{B_{5}B_{6}}\otimes I_{A_{4}B_{4}A_{5}A_{6}})$, and the operator ${\rm SWAP}_{B_{j}B_{l}}$ is a swap operator acting on registers $B_{j}$ and $B_{l}$.

To further simplify notations, we absorb all gates before and after $U_{m}$ into the density operator and measurement operator, respectively. To be more specific, let $\psi_{A_{l}B_{l}}$ denote the density operator $U_{m-1}...U_{1}\op{00}{00}_{A_{l}B_{l}}U_{1}^{\dagger}...U_{m-1}^{\dagger}$ in register $A_{l}B_{l}$, for $l=1,...,6$. And we define observable operators $K$, $O$, $G$ as follows
\begin{align}
&K=U_{m+1}^{\dagger}...U_{N}^{\dagger}(I_{A_{1}}\otimes H_{B_{1}})U_{N}...U_{m+1}.\\
&O=(U_{m+1}^{\dagger}...U_{N}^{\dagger})^{\otimes2}W_{1}(U_{N}...U_{m+1})^{\otimes 2},\\
&G=(U_{m+1}...U_{N}^{\dagger})^{\otimes 3}W_{2}(U_{N}...U_{m+1})^{\otimes 3}.
\end{align}
Then partial derivatives in Eqs.~\eqref{EQ:42}-\eqref{EQ:46} can be rewritten as {
\begin{align}
    \frac{\partial L_{1}}{\partial\theta_{m}}=&\frac{\partial}{\partial\theta_{m}}\tr(U_{m}(\theta_{m})\psi_{A_{1}B_{1}} U_{m}^{\dagger}(\theta_{m})\cdot K),\\
    \frac{\partial L_{2}}{\partial \theta_{m}}=&2\beta^{-1}\frac{\partial}{\partial\theta_{m}}\tr(U_{m}\psi_{A_{2}B_{2}}U_{m}^{\dagger}(\theta_{m})\otimes U_{m}\psi_{A_{3}B_{3}}U_{m}^{\dagger}(\theta_{m})\cdot O),\\
    \frac{\partial L_{3}}{\partial\theta_{m}}=&-\frac{\beta^{-1}}{2}\frac{\partial}{\partial\theta_{m}}\tr(U_{m}\psi_{A_{4}B_{4}}U_{m}^{\dagger}(\theta_{m})\otimes U_{m}\psi_{A_{5}B_{5}}U_{m}^{\dagger}(\theta_{m})\nonumber\\ &\otimes U_{m}\psi_{A_{6}B_{6}}U_{m}^{\dagger}(\theta_{m})\cdot G).
\end{align}}

Now, we derive the analytical forms of the derivatives of each $L_{j}$, $j=1,2,3$. Notice that the trainable unitary $U(\bm\theta)$ is a sequence of unitaries $U_{m}(\theta_{m})$ and each unitary $U_{m}(\theta_{m})=e^{-i\theta_{m}H_{m}/2}$. The partial derivative of $U(\bm\theta)$ can be explicitly given as follows,
\begin{align}
\frac{\partial U(\bm\theta)}{\partial\theta_{m}}&=U_{N}(\theta_{N})...\frac{\partial U_{m}(\theta_{m})}{\partial\theta_{m}}...U_{1}(\theta_{1}),\\
&=-\frac{i}{2}U_{N}(\theta_{N})...H_{m}U_{m}...U_{1}(\theta_{1}).\label{eq:gradient:u}
\end{align}
Using the expression of $\partial_{\theta_{m}}U(\bm\theta)$ in Eq.~\eqref{eq:gradient:u}, and some facts, including an identity $i[H_{m}, M]=U_{m}(-\pi/2)MU_{m}^{\dagger}(-\pi/2)-U_{m}(\pi/2)MU_{m}^{\dagger}(\pi/2)$, which holds true for arbitrary matrix $M$, the symmetry of the operator $O$, and an equality $\partial_{\theta_{m}}\tr(\rho(\bm\theta)^{3})=3\partial_{\theta_{m1}}\tr(\rho_{1}(\bm\theta)\otimes\rho_{2}(\bm\theta)\otimes\rho_{3}(\bm\theta)\cdot S_{1}S_{2})$ (c.f. Lemma~\ref{lemma:optimization:1} in Appendix \ref{sec:Supplemental_higher_order_gradient}), where $S_{1}={\rm SWAP}_{12}\otimes I_{3}$ and $S_{2}=I_{1}\otimes{\rm SWAP}_{23}$, the gradients of each $L_{j}$ can be estimated using the following formulas,
{
\begin{align}
\frac{\partial L_{1}}{\partial\theta_{m}}&=\frac{1}{2}(\langle K\rangle_{\theta_{m}+\frac{\pi}{2}}-\langle K\rangle_{\theta_{m}-\frac{\pi}{2}}),\\
\frac{\partial L_{2}}{\partial\theta_{m}}&=2\beta^{-1}\left(\langle O\rangle_{\theta_{m}+\frac{\pi}{2},\theta_{m}}-\langle O\rangle_{\theta_{m}-\frac{\pi}{2},\theta_{m}}\right),\\
\frac{\partial L_{3}}{\partial\theta_{m}}&=-\frac{3\beta^{-1}}{4}\left(\langle G\rangle_{\theta_{m}+\frac{\pi}{2},\theta_{m},\theta_{m}}-\langle G\rangle_{\theta_{m}-\frac{\pi}{2},\theta_{m},\theta_{m}}\right),
\end{align}}
where the notation $\langle X\rangle_{a,b}$ is defined below,
\begin{align}
\langle K\rangle_{\theta_{\alpha}}&=\tr\left(U_{\alpha}\psi_{A_{1}B_{1}} U_{\alpha}^{\dagger}\cdot K\right),\label{eq:notation:k}\\
\langle O\rangle_{\theta_{\alpha},\theta_{\beta}}&=\tr\left(U_{\alpha}\psi_{A_{2}B_{2}} U_{\alpha}^{\dagger}\otimes U_{\beta}\psi_{A_{3}B_{3}} U_{\beta}^{\dagger} \cdot O\right),\label{eq:notation:o}\\
\langle G\rangle_{\theta_{\alpha},\theta_{\beta},\theta_{\gamma}}&=\tr\left(U_{\alpha} \psi_{A_{4}B_{4}} U_{\alpha}^{\dagger}\otimes U_{\beta}\psi_{A_{5}B_{5}} U_{\beta}^{\dagger}\otimes U_{\gamma}\psi_{A_{6}B_{6}} U_{\gamma}^{\dagger} \cdot G\right).\label{eq:notation:g}
\end{align}

\subsection{Supplementary proof for Proposition~\ref{prop:circuit_with_1_para}}
\label{sec:prop:circuit_with_1_para}

\begin{proposition}
Given the circuit in Fig.~\ref{fig:Ising_ansatz_1_para} and denote $\rho_B(\theta)$ its output state on system B. For the Ising chain model, if we compute its free energy in Eq.~\eqref{eq:free_energy} and our truncated cost in Eq.~\eqref{EQ:40}, then the optimal $\theta$'s that minimize these two loss functions are both
\begin{align}
    \theta= \pi/2 + k \pi,
\end{align}
where $k\in \mathbb{Z}$. As a result, $\rho_B(\pi/2)$ is the best state, under this circuit, that approaching the Gibbs state in Eq.~\eqref{eq:Ising_Gibbs}, and their fidelity could be larger than 95\% for any $\beta \ge 1.25$.
\end{proposition}

\begin{proof}
This claim can be directly derived by computing the global minimum of both loss functions $\mathcal{F}$ and $\mathcal{F}_2$ using the ansatz in Fig.~\ref{fig:Ising_ansatz_1_para}. Assuming $n_A=1$, $n_B=n$ and denoting the output state as $\ket{\psi}_{AB}$, we can easily obtain the state $\rho_B$ as 
\begin{align}
\rho_{B}(\theta)&=\tr_{A}(\op{\psi}{\psi}_{AB})\\
&=\cos^{2}(\theta/2)\op{0^{n}}{0^{n}}_{B}+\sin^{2}(\theta/2)\op{1^{n}}{1^{n}}_{B}.
\end{align}
To compute the derivatives of $\mathcal{F}(\theta)$ and $\mathcal{F}_2(\theta)$, we first present their explicit expressions.
\begin{align}
\mathcal{F}(\theta)&=\tr(H_B\rho_{B}(\theta))-\beta^{-1}S(\rho_{B}(\theta))\nonumber\\
&=\lambda_{0}a+\lambda_{1}b+\beta^{-1}\left[a\ln(a)+b\ln(b)\right],\\
\mathcal{F}_{2}(\theta)&=\tr(H_B\rho_{B}(\theta))+\beta^{-1}\left[2\tr(\rho_{B}(\theta)^{2})-\frac{1}{2}\tr(\rho_{B}(\theta)^{3})-\frac{3}{2}\right]\nonumber\\
&=\lambda_{0}a+\lambda_{1}b+\beta^{-1}\left[2a^{2}+2b^{2}-\frac{a^{3}+b^{3}}{2}-\frac{3}{2}\right],
\end{align}
where we denote $\cos^2(\theta/2)$ by $a$ and $\sin^2(\theta/2)$ by $b$, and $\lambda_0$ and $\lambda_1$ are the eigenvalues of $H$ associated with eigenvectors $\ket{0^n}$ and $\ket{1^n}$, respectively. 

Actually, in the Ising chain model, $\lambda_0$ and $\lambda_1$ are equal (c.f. Lemma~\ref{lemma:facts}). Thus, derivatives of $\mathcal{F}(\rho_{B}(\theta))$ and $\mathcal{F}_{2}(\rho_{B}(\theta))$ with respect to $\theta$ have the following form.
\begin{align}
\partial_{\theta}\mathcal{F}(\theta)&=\frac{\beta^{-1}}{2}\sin(\theta)\left(\ln(b)-\ln(a)\right)\label{EQ:48} \\
\partial_{\theta}\mathcal{F}_{2}(\theta)&=\frac{5\beta^{-1}}{2}\sin(\theta)(b-a).\label{EQ:49}
\end{align}
From Eqs.~\eqref{EQ:48}~\eqref{EQ:49}, the global minimums of $\mathcal{F}$ and $\mathcal{F}_2$ are 
\begin{align}
    \theta=\frac{\pi}{2}+k\pi \quad \forall k\in\mathbb{Z}.
\end{align}

The fidelity between $\rho_B(\pi/2)$ and $\rho_G$ could be derived from Proposition \ref{prop:Fid_lower_bound} and Lemma \ref{lemma:facts}, where if we set $N=2^5$, $\Delta =4$ and $\beta=1.25$ and then we get $F(\rho_{B}(\pi/2),\rho_{G})\ge 95.3\%$. Hence, we could achieve a fidelity higher than 95\%, provided that $\beta\ge 1.25$.
\end{proof}

\subsection{Supplementary proof for Proposition~\ref{prop:Fid_lower_bound}}
\label{sec:proof:Fid_lower_bound}

\begin{proposition}
Given the circuit in Fig.~\ref{fig:Ising_ansatz_1_para} and let $\rho_{B}(\theta)$ be its output state on system $B$. Then the fidelity between $\rho_{B}(\pi/2)$ and the Gibbs state $\rho_{G}$ is lower bounded. To be more specific,
\begin{align}
F(\rho_{B}(\pi/2),\rho_{G})\geq\frac{1}{\sqrt{1+(N/2-1)e^{-\beta\Delta}}},
\label{EQ:fidelity:beta_appendix}
\end{align}
where $N$ is the dimension of system $B$, i.e., $N=2^{n_B}$ and $\Delta$ is the spectral gap of the Hamiltonian $H_B$ on the system $B$, i.e., the discrepancy between the minimum and the second minimum eigenvalues. 
\end{proposition}
 
\begin{proof}
To prove this result, we assume the eigenvalues, associated with the eigenvectors $\ket{0}$, $\ket{1}$, $\ldots$, $\ket{N-1}$, for the Hamiltonian $H$ are denoted by $\lambda_{0}$, $\lambda_{1}$, $\ldots$, $\lambda_{N-1}$. Specifically, eigenvalues $\lambda_{0}$ and $\lambda_{N-1}$ are associated with eigenvectors $\ket{0^{n}}$ and $\ket{1^{n}}$, respectively. A key feature of the Ising model is that $\lambda_{0}=\lambda_{N-1}$ and they are minimum among all eigenvalues and let $\Delta$ denote the spectral gap of the Hamiltonian $H_{B}$, which implies $\lambda_{j}-\lambda_{0}\geq \Delta$, where $j\neq 0,N-1$.

Let $\hat{\lambda}_{j}$ denote the eigenvalues of the Gibbs state. Then, according to the definition of Gibbs state, $\hat{\lambda}_{j}$ have the following form.
\begin{align}
\hat{\lambda}_{j}=\frac{e^{-\beta\lambda_{j}}}{Z}.
\end{align}
where $Z=\sum_{l=0}^{N-1}e^{-\beta\lambda_{l}}$.

Next, we derive bounds on eigenvalues $\hat{\lambda}_{0}$ and $\hat{\lambda}_{N-1}$. Note that, in the Ising model, eigenvalues $\lambda_{0}$ and $\lambda_{N-1}$ are equal, then the associated eigenvalues $\hat{\lambda}_{0}=\hat{\lambda}_{N-1}$, and they have the explicit forms in the following.
\begin{align}
\hat{\lambda}_{0}&=\hat{\lambda}_{N-1}=\frac{e^{-\beta \lambda_{0}}}{Z}\\
&=\frac{1}{2+\sum_{j\neq 0,N-1}e^{\beta(\lambda_{0}-\lambda_{j})}} \label{eq:experiment:lambda:1}\\
&\geq\frac{1}{2+(N-2)e^{-\beta\Delta}}.\label{eq:experiment:lambda:2}
\end{align}

Recall the output state of our algorithm is $\rho_{B}(\pi/2)$ in Proposition~\ref{prop:circuit_with_1_para}. The inequality in Eq.~\eqref{EQ:fidelity:beta_appendix} is immediately acquired by calculating the fidelity $F(\rho_{B}(\pi/2),\rho_{G})$:
\begin{align}
F(\rho_{B}(\pi/2),\rho_{G})&=\tr\sqrt{\rho_{B}^{1/2}(\pi/2)\rho_{G}\rho_{B}^{1/2}(\pi/2)}\\
&=\tr\sqrt{1/\sqrt{2}\cdot\hat{\lambda}_{0}\cdot1/\sqrt{2}\op{0^{n}}{0^{n}}+1/\sqrt{2}\cdot\hat{\lambda}_{N-1}\cdot1/\sqrt{2}\op{1^{n}}{1^{n}}}\\
&=\sqrt{2\hat{\lambda}_{0}}\\
&\geq \frac{1}{\sqrt{1+(N/2-1)e^{-\beta\Delta}}}.
\end{align}
This completes the proof.
\end{proof}

The following lemma states some facts about the Ising model, which are helpful for the above proofs.
\renewcommand{\theproposition}{S\arabic{proposition}}
\begin{lemma}\label{lemma:facts}
Given an Ising model Hamiltonian in Eq.~\eqref{ising_model}, the eigenvalues $\lambda_{0}$ and $\lambda_{N-1}$, associated with eigenvectors $\ket{0^{n_{B}}}$ and $\ket{1^{n_{B}}}$, are equal, i.e., $\lambda_{0}=\lambda_{N-1}=-L$. Particularly, the spectral gap is exactly $4$ for all $n_{B}\geq 2$.
\end{lemma}

\begin{proof}
To prove that eigenvalues of $\ket{0^{n_{B}}}$ and $\ket{1^{n_{B}}}$ are equal, we compute their corresponding eigenvalues of each term $Z_{B,i}Z_{B,i+1}$. Notice that, for all $i=1,...,L$,
\begin{align}
Z_{B,i}Z_{B,i+1}\ket{0^{n_{B}}}=\ket{0^{n_{B}}},    \\
Z_{B,i}Z_{B,i+1}\ket{1^{n_{B}}}=\ket{1^{n_{B}}}.
\end{align}
Hence, in the Ising model, the eigenvalues of $\ket{0^{n_{B}}}$ and $\ket{1^{n_{B}}}$ are $-L$.

As for the rest eigenvectors $\ket{j}$, $j\neq0,N-1$, their eigenvalues of $Z_{B,i}Z_{B,i+1}$ are given as follows.
\begin{align}
    Z_{B,i}Z_{B,i+1}\ket{j}=(-1)^{k_{i}+k_{i+1}}\ket{j},
\end{align}
where $k_{i}$ and $k_{i+1}$ are the bits in the $i$-th and $(i+1)$-th position of $\ket{j}$. Particularly, in the Ising model, the eigenvalue of $\ket{j}$ is represented as $-\sum_{i=1}^{L}(-1)^{k_{B,i}+k_{B,i+1}}$. To be specific,
\begin{align}
    H_{B}\ket{j}=-\sum_{i=1}^{L}(-1)^{k_{i}+k_{i+1}}\ket{j}.\label{eq:ising_model:eigenvalues:1}
\end{align}
Overall, the eigenvalues of $\ket{j}$ are larger than $-L$, which implies the eigenvalues of $\ket{0^{n_{B}}}$ and $\ket{1^{n_{B}}}$ are minimum.

{
Now we show that the spectral gap of the Ising model with more than two qubits is 4. The minimum eigenvalue of $H_{B}$ is $-L$ means that $k_{B,i}+k_{B,i+1}=0/2$ for all $i$, and hence the ground states are $\ket{0^{n_{B}}}$ and $\ket{1^{n_{B}}}$. If we flip one qubit of the eigenvector $\ket{j}$, then two terms like $(-1)^{k_{B,i}+k_{B,i+1}}$ of its eigenvalues will change by 2. If we flip more qubits, then more terms will change. Notice that the eigenvectors with an eigenvalue larger than $-L$ will differ from those with minimum eigenvalue at least one qubit, resulting in at least two terms change. Then, the overall difference between the minimum and second minimum eigenvalue is at least 4. Clearly, the difference of $4$ can be obtained, and then the spectral gap is $4$.}
\end{proof}

\end{document}